\documentclass[conference]{IEEEtran}
\ifCLASSINFOpdf
\else
\fi

\let\labelindent\relax

\usepackage{multirow}
\usepackage{siunitx}
\usepackage[center]{caption}
\usepackage{hyperref}
\usepackage{url}
\usepackage[english]{babel}
\usepackage[utf8]{inputenc}
\usepackage{algorithm}
\usepackage[noend]{algpseudocode}
\usepackage{amssymb,amsmath,amsthm} 
\usepackage{subcaption}
\usepackage{graphicx}
\usepackage{xspace}
\usepackage{comment}
\usepackage{xcolor}
\usepackage{thmtools, thm-restate}
\usepackage{mdwlist}
\usepackage{colortbl}
\usepackage{float}
\usepackage[export]{adjustbox}

\usepackage{nicefrac}
\usepackage{cleveref}

\makeatletter
\def\thmheadbrackets#1#2#3{%
  \thmname{#1}\thmnumber{\@ifnotempty{#1}{ }\@upn{#2}}%
  \thmnote{ {\the\thm@notefont[#3]}}}
  \thm@headfont{\itshape} 
  \thm@notefont{} 
  \itshape
\makeatother

\newtheoremstyle{brakets}
  {}
  {}
  {\itshape}
  {}
  {\bfseries}
  {.}
  { }
  {\thmheadbrackets{#1}{#2}{#3}}

\theoremstyle{brakets}

\newcounter{def}
\newtheorem{definition}[def]{Definition}

\usepackage{mathtools} 
\usepackage{paralist}
\usepackage{bbm}

\newcommand{\obs}{m}
\newcommand{\Obs}{M}
\newcommand{\weight}{w}
\newcommand{\rate}{\lambda_s}

\newcommand{\pathnum}{x}
\newcommand{\clientavg}{C^{avg}_{client}}
\newcommand{\pois}{\text{Pois}}

\newcommand{\pathnumrand}{X}
\newcommand{\obsrand}{M}
\newcommand{\capacityset}{\mathcal{C}}
\newcommand{\ourmethodclosed}{\mathit{MLEFlow\text{-}CF}}
\newcommand{\ourmethodbrute}{\mathit{MLEFlow\text{-}Q}}
\newcommand{\ourmethod}{\ensuremath{\mathit{MLEFlow}}\xspace}
\newcommand{\dualmethod}{\ensuremath{\mathit{DiProber}}\xspace}
\newcommand{\dualmethodwhole}{\ensuremath{\mathit{DiProber\text{-}WH}}\xspace}
\newcommand{\dualmethodone}{\ensuremath{\mathit{DiProber\text{-}O}}\xspace}
\newcommand{\proptorflow}{\ensuremath{\mathit{TorFlow\text{-}P}}\xspace}
\newcommand{\sbws}{\mathit{sbws}\xspace}
\newcommand{\torflow}{\mathit{TorFlow}\xspace}
\newcommand{\mlflowexp}{\mathit{MF}}
\newcommand{\dualexpwhole}{\mathit{WH}}
\newcommand{\dualexpone}{\mathit{O}}
\newcommand{\torflowexp}{\mathit{TF}}

\newcommand{\nnnreals}{\mathbb{R}^n_{\geq 0}}

\newcommand{\Var}{\mathit{Var}}
\newcommand{\Cq}{\kappa}
\usepackage[shortlabels]{enumitem}

\catcode`,\active

\catcode`\,12

\usepackage{tikz}
\usepackage{amsmath}

\usepackage{filecontents}

\begin{document}
%
\title{DiProber: Using Dual Probing to Estimate Tor Relay Capacities in Underloaded Networks}

\author{
\IEEEauthorblockN{Hussein Darir}
\IEEEauthorblockA{University of Illinois at\\ Urbana-Champaign\\
hdarir2@illinois.edu}
\and
\IEEEauthorblockN{Nikita Borisov}
\IEEEauthorblockA{University of Illinois at\\ Urbana-Champaign\\
nikita@illinois.edu}
\and
\IEEEauthorblockN{Geir Dullerud}
\IEEEauthorblockA{University of Illinois at\\ Urbana-Champaign\\
dullerud@illinois.edu
}}


%


\IEEEoverridecommandlockouts

\maketitle

\begin{abstract}
Tor is the most popular anonymous communication network. It has millions of daily users seeking privacy while browsing the internet. It has thousands of relays to route and anonymize the source and destinations of the users packets. To create a path, {\em Tor authorities} generate a probability distribution over relays based on the estimates of the capacities of the relays. An incoming user will then sample this probability distribution and choose three relays for their paths. The estimates are based on the bandwidths of observation probes the authority assigns to each relay in the network. Thus, in order to achieve better load balancing between users, accurate estimates are necessary. Unfortunately, the currently implemented estimation algorithm generate inaccurate estimates causing the network to be under utilized and its capacities unfairly distributed between the users paths. We propose $\dualmethod$, a new relay capacity estimation algorithm. The algorithm proposes a new measurement scheme in Tor consisting of two probes per relay and uses maximum likelihood to estimate their capacities. We show that the new technique works better in the case of under-utilized networks where users tend to  have very low demand on the Tor network. 

\end{abstract}


%

\section{Introduction}

Tor~\cite{tor} is the most popular anonymous communication network, with several million estimated daily users~\cite{tor-metrics-users}. It offers users a way to communicate online while preserving their identity and relationship to third parties. Tor operates by using a network of volunteer \emph{relays} to forward an encrypted version of users' traffic in order to obscure the source and/or the destination of network traffic. To ensure consistent performance, users' traffic needs to be load-balanced across the relays. The network capacities of the relays are quite heterogeneous, spanning many orders of magnitude; an accurate estimate of these capacities is an important input to the load-balancing process.

Initially Tor relied on relays' own measurements of their own capacities to generate estimates ; this, however, created the possibility of low-resource attacks on the Tor network~\cite{bauer+:wpes07}. This motivated the development of $\torflow$, a bandwidth monitoring system~\cite{torflow}. $\torflow$ uses external probes to monitor the performance of individual relays and uses this value to adjust the bandwidth value reported by the relay itself. The capacity estimates produced by $\torflow$, however, vary considerably over time and between different $\torflow$ instance, impacting the traffic allocation algorithm's ability to properly balance load across relays. A new estimation algorithm based on maximum likelihood estimation was also developed, $\ourmethod$~\cite{mleflowp}. While the estimation accuracy of this new algorithm showed a lot of promises, the  probabilistic model used to derive the estimator had many simplifying assumptions that are not necessarily true in practise. Notably, the model used assumed that users utilize all the bandwidth allocated to them while in practise users' demand fluctuate and is generally lower than the available bandwidth. Moreover, while the estimation algorithm used had significantly lower estimation error for exit relays, it still led to relatively high error for guard and middle relays. 

We propose a new maximum likelihood estimator based on the work of $\ourmethod$, where we relax the assumption aforementioned about clients usage of the available bandwidth. The developed algorithm proposes a new measurement scheme of two probes per relay and uses the results of both measurements to identify whether a relay is bottlenecked or not during each epoch. Then, depending on the case, the algorithm uses a distinct probabilistic model relating measurements to actual capacity and performs maximum likelihood estimation. We derive analytical bounds of convergence for the estimates. We then validate the results of our analysis in flow-based Python simulations, where we simulate both loaded and under-loaded networks and show the benefits of using the new measurements scheme.

\section{Path allocation in Tor}
\label{sec:true_tor_model}


The current Tor network consists of around $6000$ {\em relays}~\cite{tor-metrics-servers} that are used to forward user traffic. To create a connection, a user chooses a \emph{path} of three different relays to construct a circuit that forwards traffic in both directions. Only the user knows the entire path; the relays know only their predecessor and successor, obscuring the relationship between clients and destinations. The traffic is also encrypted / decrypted at each node to hide the correspondence between incoming and outgoing traffic from a network observer.

\begin{figure}
    \includegraphics[width=\columnwidth]{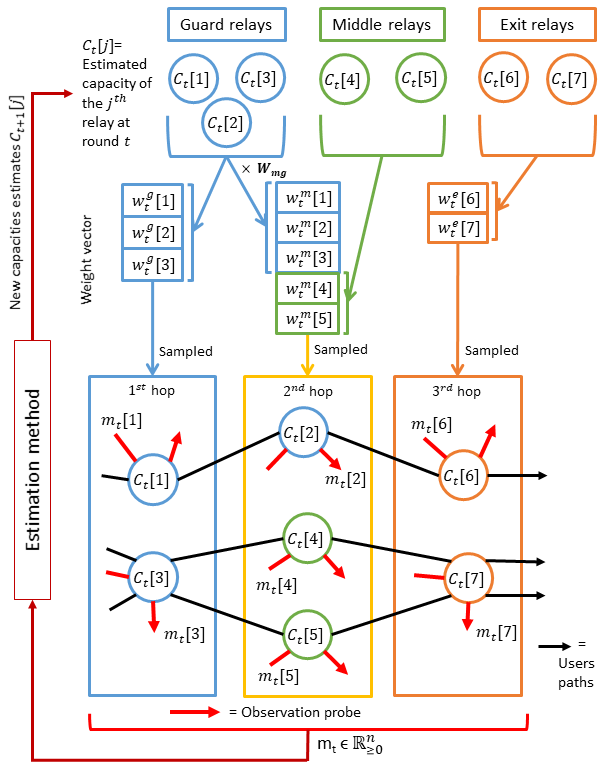}
    \caption{\small Tor relay selection and capacity estimation. Users select three relays from collections $G$, $G\cup M$, and $E$ using respective weight vectors $w_t^g, w_t^m$, and $w_t^e$, to 
    form a \emph{path} (black arrows). The bandwidth authority collects measurements $\obs_t$ of each relay (red arrows), and updates the capacity estimates $C_{t+1}$ in the next consensus document,
    which are then used to generate new weight vectors $w^x_{t+1}$.}    \label{fig:general_problem_definition}
\end{figure}

%
 
 Relays in Tor have heterogeneous capabilities and have network capacity\footnote{By ``capacity'' we refer to the smaller of upload and download bandwidth limit on the relay. This may be imposed by the ISP, the network configuration, or manually configured by the relay operator. In some cases, there may exist other bottlenecks on the path between two relays but a per-node bandwidth limit is a common and useful model of network capacity constraints.}
 sizes that differ by orders of magnitude (see Figure~\ref{fig:100_network_capacities}). Relays also have different capabilities and can be divided into three classes: \emph{exits}, which can be used in the last position of the path, \emph{guards}, which can be used in the first or second position, and \emph{middles} which can only be used in the second position~\cite{tor-dir-spec}. We denote the corresponding sets of relays by $E$, $G$, and $M$, respectively. To create a path, nodes are sampled from these sets with a probability proportional to their estimated capacity. For example, if we define $C[j]$ to be the estimated capacity of relay $j$, then the probability of choosing relay $j \in E$ as the last node in a path is $w^e[j] = C[j] / \left(\sum_{j' \in E} C[j']\right) $; likewise for guard nodes being chosen in the first position. The middle position can be chosen from both guard and middle nodes; to balance bandwidth among classes, guard node capacity is adjusted by a multiplier $W_{mg}$; i.e., a guard node $j \in G$ is chosen for the middle position with probability: 
 \[ w^m[j] = \frac{W_{mg} C[j]}{\sum_{j' \in G} W_{mg} C[j'] + \sum_{j' \in M} C[j']} \]
 The multiplier is computed as:
 \[ W_{mg} = \frac{\sum_{j' \in G} C[j'] - \sum_{j' \in M} C[j]}{2 \sum_{j' \in G} C[j']} \]
 
 This is a somewhat simplified presentation that describes the scenario where exit bandwidth is scarce and there is more guard bandwidth than middle bandwidth, as is the case in the actual Tor network. (See the Tor Directory Specification for more details on how other cases would be handled~\cite{tor-dir-spec}.) See Figure~\ref{fig:general_problem_definition} for a description of this process.

 It is easy to see that, in this scenario, if the estimated capacities are equal to the true relay capacities, which we will call $C^*[j]$, the expected number of paths using each  exit relay will be proportional to its bandwidth; likewise, the expected number of paths using each guard and middle node will be proportional to their bandwidth. Using $X[j]$ to denote the number of paths on relay $j$, we have:
 \begin{align*}
     E[X[j]]/C^*[j] =E[X[j']]/C^*[j'] \quad & \text{for $j, j' \in E$} \\
     E[X[j]]/C^*[j] = E[X[j']]/C^*[j'] \quad & \text{for $j, j' \in G \cup M$}
 \end{align*}
Thus, in expectation, each path would have the same bandwidth---$C^*[j] / E[X[j]]$ for $j \in E$. Our goal is therefore to estimate these capacities as accurately as possible.\footnote{Note that some research suggests allocation other than proportional to bandwidth results in better performance~\cite{snader-borisov:tdsc11,herbert2014optimising}; nevertheless, an accurate capacity estimate is still needed for these alternative  path allocation strategies.}

\subsection{Capacity Estimation}
\label{sec:capest}

Each relay estimates its own network capacity by computing the maximum sustained download and upload bandwidth over a 5-second period over the last 5 days. It reports this value (called the \emph{observed bandwidth}) to directory authorities, who then compile it across all relays and distribute the information to the clients in a \emph{consensus} document, published every hour. We will use $b_t[j]$ to refer to the observed bandwidth of relay $j$ in the consensus document published at time $t$. 

Using the observed bandwidth directly for load-balancing creates the opportunity for a low-resource attack on the Tor network~\cite{bauer+:wpes07}. In particular, a relay can publish a high observed bandwidth for itself, which will cause more clients to choose it, and create more chances for it to break users anonymity. This motivated the design of $\torflow$~\cite{torflow}, which used observations of actual relay performance to estimate capacities, rather than simply trusting the value reported by the relay itself. In $\torflow$, a \emph{bandwidth authority} creates probe circuits through each relay and downloads a file of a certain size, measuring the realized bandwidth.\footnote{Since Tor does not allow one-hop circuits, these circuits use two relays: the relay under measurement and a high-bandwidth relay.} We will call this the \emph{measured bandwidth}, $\obs_t[j]$. Note that in a perfectly load-balanced network, all of these observations should be equal, regardless of the chosen relay. The design of $\torflow$ uses a PID controller to attempt to bring these observations into balance.

More specifically, let $C^\torflowexp_t[j]$ be the  capacity of relay $j$, as estimated by $\torflow$, at time $t$. $\torflow$ computes an error term, $e_t[j]$ as the difference of the measured bandwidth and the average measured bandwidth, normalized by the average bandwidth:
\[ e_t[j] = (\obs_t[j] - \bar{\obs_t})/\bar{\obs_t} \]
\noindent where \[ \bar{\obs_t} = \sum_{j \in G \cup M \cup E} \obs_t[j] / \left(|G|+|M|+|E|\right) \]
It then computes the new estimate as~\cite[\S 3.1]{torflow-spec}:
\[ C^\torflowexp_{t+1}[j] = C^\torflowexp_t[j] \left(1 + K_p e_t[j] + K_i \int_{0}^t e_{t'}[j] dt' + K_d \frac{d e_t[j]}{dt} \right) \]
The constants $K_p, K_i,$ and $K_d$ control the proportional, integral, and derivative components of the PID controller. In the default configuration of $\torflow$, $K_i = K_d = 0$ and $K_p = 1$, so we will call this version of $\torflow$ \proptorflow. In this case the update equation can be simplified as:
\[ C^\torflowexp_{t+1}[j] = C^\torflowexp_t[j] \obs_t[j] / \bar{\obs_t} \]


Since the update equation does not have a normalization step, when \proptorflow was enabled in late 2011 in the actual Tor network, the absolute values of estimated bandwidth grew without bound%
\footnote{As can be seen on this graph: \url{https://metrics.torproject.org/totalcw.html?start=2011-06-01&end=2011-12-31}}.  This caused the Tor network to turn off \proptorflow and switch to using a version of $\torflow$ that uses adjusted observed bandwidth instead, called $\sbws$. We will call estimates produced by this version $C^A$, with:
\[ C^A_{t+1}[j] = b_t[j] \obs_t[j] / \bar{\obs_t} \] 
This version uses the observed bandwidth published by the relay itself, but adjusts it down if the observed bandwidth is below average or up if it is above. It has been in use in Tor since 2012
; however, it has a number of disadvantages. A relay that is not sufficiently loaded may underestimate its observed bandwidth; this leads to a well-documented ramp-up period of new relays, where their low observed bandwidth leads to a small estimated capacity and low load, which in turn leads to low observed bandwidth~\cite{tor-relay-lifecycle}.
But even established relays see their observed bandwidth change. Figure~\ref{fig:observed_bw} shows the observed bandwidth of 10 randomly selected relays over the month of May 2020, demonstrating that the observed values vary significantly over time. 


\begin{figure}
\centering
    \includegraphics[width=0.7\columnwidth]{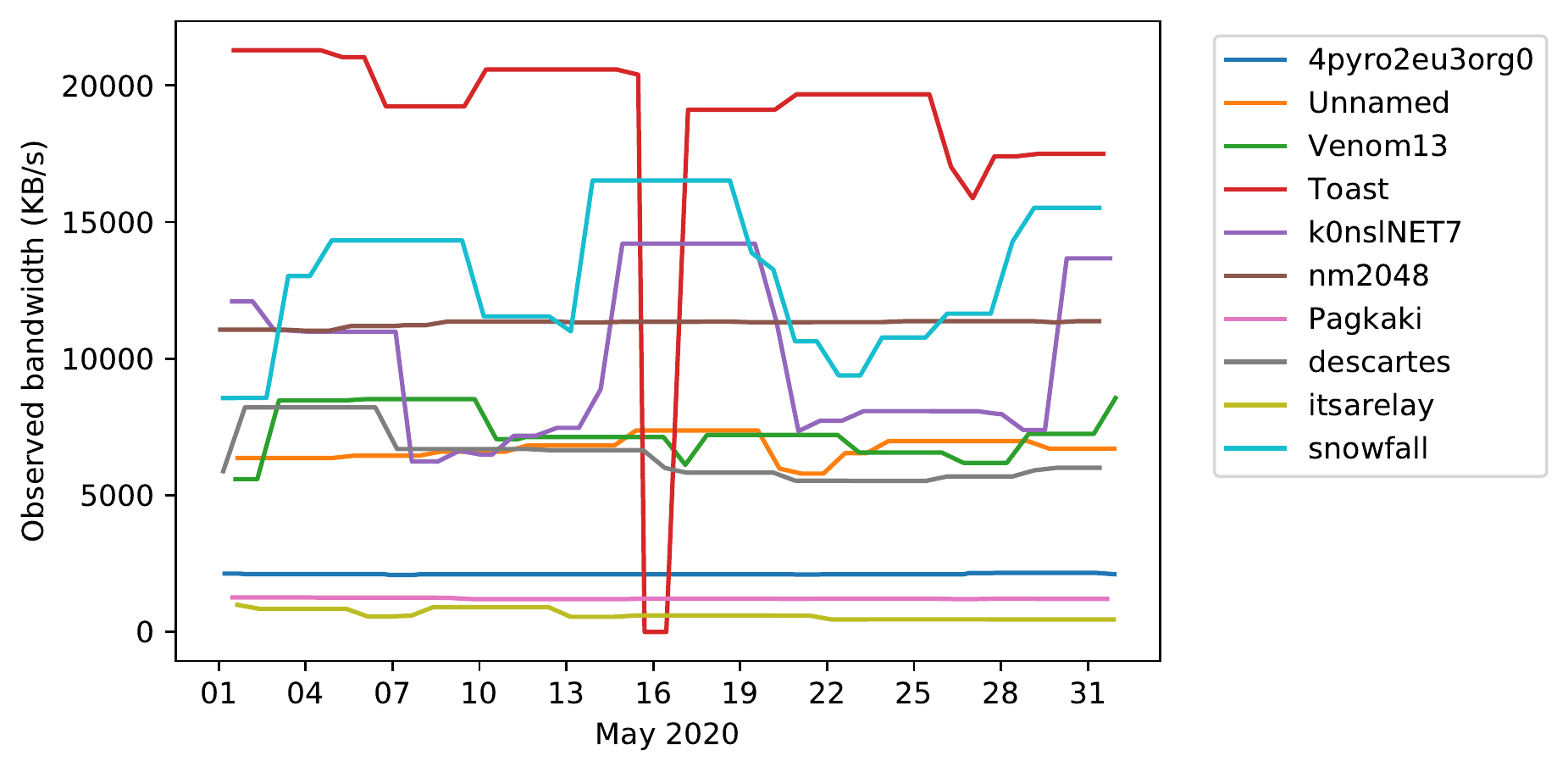}
    \hfill
    \caption{Variation in observed bandwidth in Tor relays over the month of May 2020: plot of 10 randomly selected relays.} 
    \label{fig:observed_bw}
\end{figure}



A new method for estimating the capacity of relays based on maximum likelihood estimation was proposed with $\ourmethod$~\cite{mleflowp}. To develop $\ourmethod$, a probabilistic model relating the actual relay capacities $C^*[j]$'s and the bandwidth measurements $\obs[j]$'s was derived. In order to derive this relationship, a number of assumptions was made in \cite{mleflowp}:
\begin{enumerate}
 \item  
 Relays fall into a single category and each user path goes through only a single relay. 
 \item A synchronized model where time is divided into epochs and user connections all terminate at the end of each epoch.
 At the end of the epoch the weight vector is updated and the new vector is used by all users in the next epoch.
 \item Users arrive to the network randomly following a Poisson process with rate $\rate$, denoted by $\pois(\rate)$.
\item Clients circuits use all the bandwidth allocated to them and are only bottlenecked by the relays.
 \end{enumerate}
The estimates produced by this algorithm, denoted $C^H$ are then computed using: 
\[ C_{t+1}^H[j] = \underset{\Cq \in \capacityset}{\operatorname{argmax}}\ f(\Cq, \obs_{[t]}[j], \weight_{[t]}[j]), \text{ where}\]
\[f(\Cq, \obs_{[t]}[j], \weight_{[t]}[j])=  \prod_{i=0}^{t} \frac{e^{-\rate w_i[j]}}{\big(\frac{\Cq}{\obs_i[j]}-1\big)!}(\rate w_i[j])^{\frac{\Cq}{\obs_i[j]}-1}\]
While the algorithm derived showed a lot of promises, the model used to derive the estimator assumed that users can utilize arbitrary amounts of bandwidth and are only bottlenecked on the Tor network. In practice, at times the client demand on Tor is lower than the overall available bandwidth. When simulated in a low-load scenario, $\ourmethod$ tended to misestimate relays capacities. Also, due to this assumption, the probabilistic model in ~\cite{mleflowp} is expected to produce useful results for exit relays since those relays are generally scarce in the Tor network and are expected to be paths bottleneckes. This is not true for guard and middle relays. As shown in both Python and Shadow simulations, guard and middle relays had larger estimation errors when using $\ourmethod$.

\section{$\dualmethod$: Maximum Likelihood Estimation of Relays Capacities using dual probing}
\label{sec:theory}

We propose a new method $\dualmethod$ for estimating the capacity of relays based on maximum likelihood estimation (MLE). We use the same assumptions used to derive $\ourmethod$ while relaxing the last assumption made. The paths in our model consist of single relays. We assume that time is divided into epochs and connections terminate at the end of each epoch.  We denote the number of paths passing through the $j^{\mathit{th}}$ relay in the $i^{\mathit{th}}$ epoch by $\pathnum_i[j]$. Hence, given $\weight_i$, the number of paths using the $j^{\mathit{th}}$ relay during the $i^{\mathit{th}}$ epoch is a random variable $\pathnumrand_i[j]$ with distribution $\pois(\rate\weight_i[j])$. We design an updated measurement mechanism based on two measurement probes assigned by the authorities to each relay instead of just one. The result of the two measurements is used to group the relays in two different groups, then derive a probabilistic model for each case and apply MLE. 
In this section, we will present the new measurement scheme proposed as well as the probabilistic models derived to relate the actual capacity of a relay to the measurements obtained. We then derive analytical guarantees to the new estimation algorithm proposed.

\subsection{Measurement mechanism}
\label{sec:meas_mech}
Currently, {\em Tor authorities} assign a measurement probe to each relay in the network, download a file of a certain size and measure the realized bandwidth of the probe. Instead of having only one measurement probe for each relay, we propose having two measurement probes added sequentially to each relay. The authorities start by activating the first probe and measure its realized bandwidth, denoted $\obs^1_i[j]$. Then, while the first probe is still active, the authority starts the second probe circuit and measure its bandwidth, denoted $\obs^2_i[j]$. 

Using both \emph{non-noisy} measurements of a  relay $j$ at a given epoch $t$, we can divide the relays in two groups and derive a probabilistic model for each one. First, we let $C_{client,i}[j]$ be the total bandwidth used by clients using relay $j$ during epoch $i$. 

\textbf{Case 1: Relay is not bottlenecked by clients \underline{and} when adding the two probes.} 

When the clients using relay $j$, during the $i^{th}$ epoch, are not using all the available bandwidth of the relay we have $C_{client,i}[j] < C^*[j]$. After adding the first probe, in the case where the capacity of the relay was equally divided between the users and the probe, each path using relay $j$ will have a bandwidth of $\frac{C^*[j]}{\pathnum_i[j] + 1}$. Hence all the clients using this relay will have a total bandwidth of $\pathnum_i[j]\frac{C^*[j]}{\pathnum_i[j] + 1}$. If $C_{client,i}[j] < \pathnum_i[j]\frac{C^*[j]}{\pathnum_i[j] + 1}$, then client utilization is not affected by adding the first probe and the probe will use all the remaining \emph{unused capacity} of the relay, notably $\obs^1_i[j] = C^*[j]-C_{client,i}[j]$. The same logic is true when adding the second probe. If $C_{client,i}[j] < \pathnum_i[j]\frac{C^*[j]}{\pathnum_i[j] + 2}$ then client utilization is not affected by adding the second probe. The remaining \emph{unused capacity} will then be equally divided between the two probes, notably $\obs^2_i[j]=\frac{C^*[j]-C_{client,i}[j]}{2} = \frac{\obs^1_i[j]}{2}$. Figure~\ref{fig:not_bottlenecked} illustrate the idea aforementioned.%
\begin{figure}
    \centering
    \includegraphics[width=0.9\columnwidth]{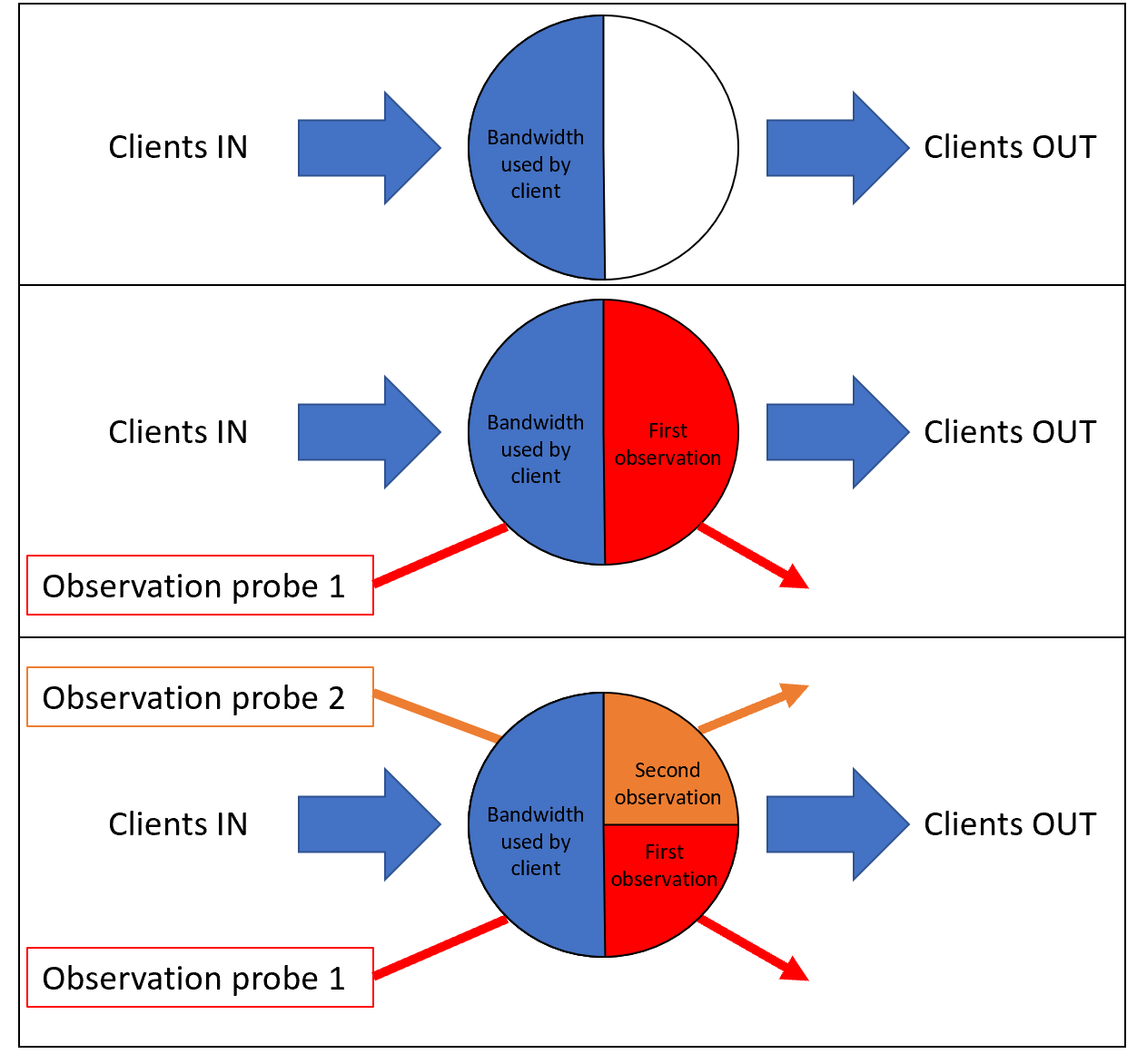}
    \caption{\textbf{(Case 1)} Relay not bottlenecked by users.}
    \label{fig:not_bottlenecked}
\end{figure}

\textbf{Case 2: Relay bottlenecked by clients \underline{or} when adding any of the two probes.}
\begin{enumerate}[(a)]
    \item If the clients are using all the available capacity of the relay there will be no remaining unused capacity. Hence, as in \cite{mleflowp}, the capacity of the relay will be divided equally between all the paths going through it, notably $\obs^1_i[j]=\frac{C^*[j]}{\pathnum_i[j]+1}$ and $\obs^2_i[j]=\frac{C^*[j]}{\pathnum_i[j]+2}$. Figure~\ref{fig:bottlenecked} illustrates this case.
    \item In this case, the clients are not using all the available capacity of relay. However, when adding the first probe, the client utilization can be affected if $C_{client,i}[j] > \pathnum_i[j]\frac{C^*[j]}{\pathnum_i[j] + 1}$. In other words, clients will then be using all the capacity available to them after the first probe was added and there will be no remaining unused capacity. Thus as case (a), the capacity of the relay will be divided equally between all the paths going through it, notably $\obs^1_i[j]=\frac{C^*[j]}{\pathnum_i[j]+1}$ and $\obs^2_i[j]=\frac{C^*[j]}{\pathnum_i[j]+2}$.
    \item In this case, clients are not using all the available capacity of the relay, nor adding the first probe will affect their utilization. Thus, the first probe will be first assigned all the remaining \emph{unused capacity}, $\obs^1_i[j] = C^*[j]-C_{client,i}[j]$. However, when adding the second probe, the clients utilization is affected if $C_{client,i}[j] > \pathnum_i[j]\frac{C^*[j]}{\pathnum_i[j] + 2}$ and there will be no remaining capacity. Thus, $\obs^2_i[j]=\frac{C^*[j]}{\pathnum_i[j]+2}$.
\end{enumerate}
\begin{figure}
    \centering
    \includegraphics[width=0.9\columnwidth]{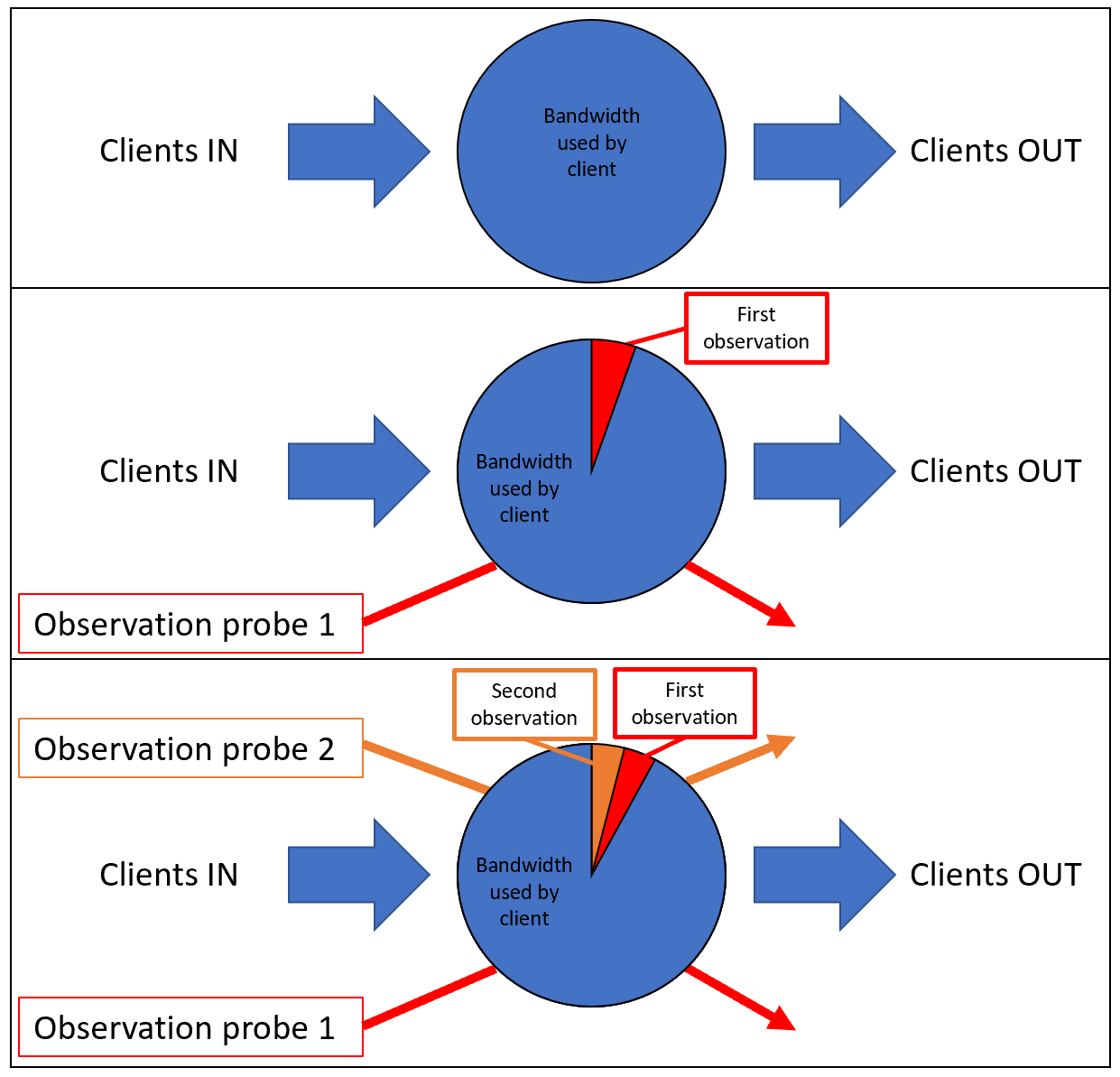}
    \caption{\textbf{(Case 2)} Relay bottlenecked by users.}
    \label{fig:bottlenecked}
\end{figure}

Thus, if the relay falls under case 1, the second observation of the relay will be equal to half the first observation obtained and the relationship between the actual capacity of the relay $C^*[j]$ and the observation is  $\obsrand^2_i[j]=\frac{C^*[j]-C_{client,i}[j]}{2}$. While for relays that fall under the second case, we can't make the same conclusion about $\obs^1_i[j]$ and $\obs^2_i[j]$. However, for all subcases of case 2, we have a relationship between $C^*[j]$ and $\obsrand^2_i[j]$ with $\obsrand^2_i[j]=\frac{C^*[j]}{\pathnumrand_i[j]+2}$.

\subsection{MLE capacities estimation}
\label{sec:MLE_formulation}

In this section, we will show the method used to compute the maximum likelihood estimation of relays capacities using the sequence of pairs of non-noisy measurements and the weights published by the Tor authority. 
More specifically, for any relay $j \in [n]$, the MLE estimate of its actual capacity $C^*[j]$ is the maximizer in $\capacityset \subset \nnnreals$ of the probability of observing the full history of measurements $(\obs^1_{[t]}[j], \obs^2_{[t]}[j])$, given the published weights $\weight_{[t]}[j]$ over the first $t+1$ periods. 
First at each epoch $i \in [t]$, and in order to use the correct probabilistic model, the algorithm compares the values of $\obs^1_i[j]$ and $\obs^2_i[j]$ and determine which case of the two cases discussed in \ref{sec:meas_mech} is true for the relay $j$ at the $i^{th}$ epoch. 
\begin{itemize}
    \item If $\obs^1_i[j] = 2\obs^2_i[j]$, then $\obsrand^2_i[j]=\frac{C^*[j]-C_{client,i}[j]}{2}$.
    \item If $\obs^1_i[j] \neq 2\obs^2_i[j]$, then $\obsrand^2_i[j]=\frac{C^*[j]}{\pathnumrand_i[j]+2}$
\end{itemize}

Assuming that we know total client utilization of each relay during each epoch is not a practical assumption. In this paper, we assume that the average utilization of a client on the Tor network is known and denoted $\clientavg$. Thus the total utilization of a relay $j$ during the $i^{th}$ epoch will be $\pathnumrand_i[j]\clientavg$

We add the superscript $D$ to the capacity estimate to denote that the result of two probes is considered. 

\begin{restatable}[MLE estimates using dual probing]{theorem}{mledual}
 \label{thm:mledual}
For any $j \in [n]$ and $t \in \mathbb{N}$,  the MLE estimate of $C^*[j]$ given the weight and observation pairs vectors $\weight_{[t]}[j]$ and $(\obs^1_{[t]}[j], \obs^2_{[t]}[j])$ is
\begin{align}
\label{eq:mle_full_history_def}
C_{t+1}^D[j] = \underset{\Cq \in \capacityset}{\operatorname{argmax}}\ \prod_{i=0}^{t}f(\Cq, \obs_{i}[j], \weight_{i}[j]), \text{ where}
\end{align}
$f(\Cq, \obs_{i}[j], \weight_{i}[j]) = $

\resizebox{0.45\textwidth}{!}{$
 \left\{
 \begin{array}{ll}
 \Pr_{\pathnumrand_{i}[j] \sim  \pois ( \rate w_{i}[j])} \left( \frac{\Cq-\pathnumrand_{i}[j]\clientavg}{2}=\obs^2_{i}[j]\right) \text{if} \hspace{0.1cm}\obs^1_i[j] = 2\obs^2_i[j] \\
 \Pr_{\pathnumrand_{i}[j] \sim  \pois ( \rate w_{i}[j])} \left( \frac{\Cq}{\pathnumrand_{i}[j] + 2}=\obs^2_{i}[j]\right) \text{if}\hspace{0.1cm} \obs^1_i[j] \neq 2\obs^2_i[j] 
 \end{array}
 \right.$}

We now write the probability in equation~(\ref{eq:mle_full_history_def}) in terms of the known quantities: $\lambda_s$, $\weight_{[t]}$, and $\obs_{[t]}$. 
%
For any $j \in [n]$ and $t \in \mathbb{N}$ the function $f$ of equation~(\ref{eq:mle_full_history_def}) can be written as follows:\\
$f(\Cq, \obs_{i}[j], \weight_{i}[j]) = $
\begin{align}
\label{eq:f_equation}
\resizebox{0.45\textwidth}{!}{$
 \left\{
 \begin{array}{ll}
 \exp \left( - \rate w_{i}[j] \right)\frac{1}{\left(\frac{\Cq-2\obs_i^2[j]}{\clientavg}\right)!}\left( \rate w_{i}[j]\right)^{\frac{\Cq-2\obs_i^2[j]}{\clientavg}}, \text{if} \hspace{0.1cm}\obs^1_i[j] = 2\obs^2_i[j] \\
 \\
 \frac{\exp \left(-\rate w_i[j]\right)}{\big(\frac{\Cq}{\obs^2_i[j]}-2\big)!}(\rate w_i[j])^{\frac{\Cq}{\obs^2_i[j]}-2}, \text{if}\hspace{0.1cm} \obs^1_i[j] \neq 2\obs^2_i[j] 
 \end{array}
 \right.$}
\end{align}
\end{restatable}

\begin{definition}[$\dualmethodwhole$]
\label{def:dualwhole}
For any relay $j \in [n]$ and $t \in \mathbb{N}$, $\dualmethodwhole$ updates the weight vector by: 
\begin{align}
\label{eq:weightwhole}
    \weight_{t+1}^{\dualexpwhole}[j]=\frac{C_{t}^{D}[j]}{\sum_{k=0}^{n} C_{t}^{D}[k]}.
\end{align}
where we use the supersctipt $\dualexpwhole$ in $\weight_{t}^{\dualexpwhole}$ to identify $\dualmethodwhole$.
\end{definition}
Equation~(\ref{eq:mle_full_history_def}) does not account for non-modeled noise in the measurements. Once noise affects $\obs^1_i[j]$ and $\obs^2_i[j]$, we will not be able to accurately discern between the two cases. 

\subsection{One step maximum likelihood}
\label{sec:one_step_MLE}

In this section, we look at the special case of only considering the last relay measurement at each epoch instead of the whole history of measurements. We add the superscript $D1$ to the capacity estimate to denote that only the last result of the two probes is considered. Equation~\ref{eq:mle_full_history_def} can now be written as,
\begin{align}
\label{eq:mle_one_def}
C_{t+1}^{D1}[j] = \underset{\Cq \in \capacityset}{\operatorname{argmax}}\ f(\Cq, \obs_{t}[j], \weight_{t}[j]), \text{ where}
\end{align}
$f(\Cq, \obs_{t}[j], \weight_{t}[j])$ is defined in Equation~\ref{eq:f_equation}. 

To approximate the maximizer of equation~(\ref{eq:mle_full_history_def}), we find $\Cq$ that sets the derivative of $f$ to zero. The result is given in the following theorem. 
\begin{restatable}[One step MLE closed form]{theorem}{oneobsclosedform}
 \label{thm:oneobsclosedform}
 For any $j \in [n]$ and $t \in \mathbb{N}$,  the MLE estimate of $C^*[j]$ given the last weight and observations pair, $\weight_{t}[j]$ and ($\obs^1_{t}[j], \obs^2_{t}[j]$ is
 \begin{align}
C_{t+1}^{D1}[j] &= \left\{
 \begin{array}{ll}
 \rate \weight_{t}[j]\clientavg + 2\obs_t^2[j], \text{if} \hspace{0.1cm}\obs^1_i[j] = 2\obs^2_i[j] \\
 \\
 \obs_t^2[j]\left(\rate \weight_{t}[j]+ 2 \right), \text{if}\hspace{0.1cm} \obs^1_i[j] \neq 2\obs^2_i[j] 
 \end{array}
 \right.
\label{eq:onecls_form}
\end{align}
\end{restatable}
From Theorem~\ref{thm:oneobsclosedform}, in the case where the relay is not bottlenecked, i.e. $\obs^1_i[j] = 2\obs^2_i[j]$, the estimated capacity is equal to the sum of the unused capacity, $2\obs^2_t[j]$, and the expected capacity used by clients, $\rate w_{t}[j]\clientavg$.
\begin{center}
    $C_{t+1}^{D1}[j]= \underbrace{\rate \weight_{t}[j]\clientavg}_\text{expected capacity used by clients} + \underbrace{\obs_t^1[j]}_\text{capacity unused by clients}$
\end{center}
While in the case where the relay is bottlenecked, the estimated capacity is the result of the product of the observation, $\obs^2_t[j]$, and the expected number of users, $\rate\weight_t[j]$. The added two in the denominator refers to the two probes added.  
\begin{center}
    $C_{t+1}^{D1}[j]= \obs_t^2[j]\left(\underbrace{\rate \weight_{t}[j]}_\text{expected number of users} + 2 \right)$
\end{center}
\begin{definition}[$\dualmethodone$]
\label{def:dualone}
For any relay $j \in [n]$ and $t \in \mathbb{N}$, $\dualmethodone$ updates the weight vector by: 
\begin{align}
\label{eq:weightone}
    \weight_{t+1}^{\dualexpone}[j]=\frac{C_{t}^{D1}[j]}{\sum_{k=0}^{n} C_{t}^{D1}[k]}.
\end{align}
where we use the superscript $\dualexpone$ in $\weight_{t}^{\dualexpone}$ to identify $\dualmethodone$.
\end{definition}

\subsection{Convergence of $\dualmethodwhole$ and $\dualmethodone$ estimates}
\label{sec:convergence_in_mean}

In this section we show that, starting with any initial weight vector, the mean of the dual probing algorithm estimates for any relay capacity, whether considering the full history in every update as in $\dualmethodwhole$ (Definition~\ref{def:dualwhole}) or only the most recent measurement in every update as in $\dualmethodone$ (Definition~\ref{def:dualone}), converges to the actual relay capacity. We also show that the variance of the estimates when using $\dualmethodwhole$ estimation goes to zero as the number of epochs increases.

\begin{restatable}[Estimates of both methods converge]{theorem}{convguarantees}
\label{thm:conv_gaurantees}
For any $j \in [n]$, $t \in \mathbb{N}$, and a method $y \in \{\proptorflow,\ourmethodclosed\, \dualmethodwhole \}$,
\begin{align}
\label{eq:conv_upperbound}
&\mathbb{E}[C_t^{y}[j]] \leq C^*[j].\\
\label{eq:conv_bound}
\mbox{Moreover, as $t \rightarrow \infty$, }   
&\mathbb{E}[C_{t}^{y}[j]] \geq C^*[j]\left(1 - \frac{1}{\rate \weight^*[j]}\right).
\end{align} 
\end{restatable}
\begin{restatable}[More users paths leads to a better convergence]{corollary}{convergenceInMean}
As the rate of users arrival $\lambda_s \rightarrow \infty$, for any $j \in [n]$, $t \in \mathbb{N}$, and method $y \in \{\dualmethodone,\dualmethodwhole\}$,
$\mathbb{E}[C_{t}^{y}[j]] \rightarrow C^*[j]$.
\end{restatable}

Furthermore, we show that the variance of the estimates of $\dualmethodwhole$ converge to zero.  This shows that $\dualmethodwhole$ provides stable and consistent estimates.
\begin{restatable}
[Convergence of variance of $\dualmethodwhole$]{theorem}{varconvergence}
As $t \rightarrow \infty$, $\Var[C_{t}^{\dualexpwhole}[j]] \rightarrow 0$.
\end{restatable}

Hence the variances of $\dualmethod$ estimates when considering the full history of weights and measurements are close to the actual capacities. 
 
We will show this experimentally in the next section.

\subsection{$\dualmethod$: Implementation}

In order to solve the maximization of Theorem~\ref{thm:mledual}, we discretize the bounded capacity set $\capacityset$ and iteratively find the maximizer of~(\ref{eq:mle_full_history_def}).  Note that quantization requires knowing a lower and upper bound on the relay capacity, which can be estimated based on past observations.

Consider a partition $\bar{\capacityset}$ of $\capacityset$ into bins. The set $\bar{\capacityset}$ contains the centers of the bins of $\capacityset$.
For any $j \in [n]$, $t \in \mathbb{N}$, and $\Cq \in \Bar{\capacityset}$, we define ${L_t}(j,\Cq)$ to be: 

\begin{align}
\label{eq:L_equation}
\resizebox{0.45\textwidth}{!}{$
 \left\{
 \begin{array}{ll}
 - \rate w_{t}[j] - \log \left(\left(\frac{\Cq-\obs_t^1[j]}{C_{avg,t}}\right)!\right)+\left( \frac{\Cq-\obs_t^1[j]}{C_{avg,t}}\right) \log \left(\rate w_{t}[j] \right), \text{if} \hspace{0.1cm}\obs^1_i[j] = 2\obs^2_i[j] \\
 \\
 - \rate w_{t}[j] - \log \left(\left(\frac{\Cq}{\obs_t^2[j]}-2\right)!\right)+\left( \frac{\Cq}{\obs_t^2[j]}-2\right) \log \left(\rate w_{t}[j] \right), \text{if}\hspace{0.1cm} \obs^1_i[j] \neq 2\obs^2_i[j] 
 \end{array}
 \right.$}
\end{align}

the $t^{\mathit{th}}$ term of the sum when taking the $\log$ of~(\ref{eq:f_equation}).

For $\dualmethodone$, since we are only considering the last measurement, the estimate of the capacity of relay $j$ is computed by iteratively searching for the maximizer $\Cq$ over the discretized capacity set $\bar{\capacityset}$ in the following equation: 
\begin{align}
C_{t+1}^{\dualexpone}[j] := \max_{\Cq \in \Bar{\capacityset}} {L_{t}}(j,\Cq).
\end{align}

For $\dualmethodwhole$, the sum of ${L_t}(j,\Cq)$ over measurement periods is stored in a variable $S_{t+1}(j,\Cq)$:
\begin{center}
    ${S_{t+1}}(j,\Cq)={S}_{t}(j,\Cq) + {L_t}(j,\Cq)$,
\end{center}
with ${S}_{0}(j,\Cq) =0$.
Then, the maximum likelihood estimate is computed by iteratively searching for the maximizer $\Cq$ over the discretized capacity set $\bar{\capacityset}$ in the following equation:
\begin{align}
C_{t+1}^{\dualexpwhole}[j] := \max_{\Cq \in \Bar{\capacityset}} {S_{t+1}}(j,\Cq).
\end{align}

\section{Analytical Model Simulations: Low Fidelity  Experiments}
\label{sec:Py_results}

To better understand the properties and performance of $\proptorflow$, $\ourmethod$, $\sbws$ and $\dualmethod$, we evaluated them using simulations of  our analytical model of the Tor network. These simulations are implemented in Python;  we will make our implementation publicly available at publication time.

We evaluate the performance of the different methods using two metrics: (a) the accuracy of the relay capacity estimates in a network where clients use all the bandwidth allocated to them, (b) the accuracy of the relay capacity estimates in an underloaded network and (c) the amount of bandwidth allocated to the user paths resulting from the weight vectors generated using the capacity estimates.

The simulation algorithm we have used is shown in  Algorithm~\ref{alg:low_experiments}. The algorithm takes as input: the Poisson arrival rate of users $\lambda_s$,  the total number of measurement periods $T$ to be simulated, 
a {\em method} $ \in \{\mathit{Actual}, \proptorflow,$ $\ourmethod, \dualmethodone, \dualmethodwhole\}$ to compute the capacities of the relays from measurements.  The algorithm outputs the bandwidths allocated for users paths and the weight vectors published over all periods between $0$ and $T$.

The simulation algorithm iterates over measurement periods. In each period $i$, it generates the total number of users paths $N_i$ that will join the network by sampling a Poisson distribution with rate $\lambda_s$ in line~\ref{ln:pick_num_paths}. Then, it uses the weight vector $w_i$ computed in the previous period as a probability distribution for the users to choose the relays of their paths from in line~\ref{ln:pick_paths}. In line~\ref{ln:compute_obs1}, it adds the first probe to each relay and uses the max-min fairness bandwidth allocation algorithm to get the bandwidth allocated for each path, and thus generate the observation vector $\obs^1_i$. In line~\ref{ln:compute_obs2}, it adds the second probe and generates $\obs^2_i$. After that, it computes $w_{i+1}$ using the given {\em method} in line~\ref{ln:update_w}. Finally, it deletes all the paths for a fresh start of the next period.   



 
 \begin{algorithm}
\caption{Low fidelity simulation}
\label{alg:low_experiments}
\begin{algorithmic}[1]
\State \textbf{input:} $\lambda_s, T$, $\mathit{method} \in \{\mathit{Actual}$, $\proptorflow, \ourmethod, \dualmethodwhole, \dualmethodone\}$.
\For {$i \in [0,...,T]$}
\State Pick the number of users $N_i \sim Poi(\lambda_s)$. \label{ln:pick_num_paths}
\State Construct users paths of $3$ relays using $w_i$. \label{ln:pick_paths}
\State Compute ${\obs^1_i}$ using max-min bandwidth allocation. \label{ln:compute_obs1}
\State Compute ${\obs^2_i}$ using max-min bandwidth allocation. \label{ln:compute_obs2}
\State Compute $w_{i+1}$ based on $\obs_i$ and $w_i$ using $\mathit{method}$. \label{ln:update_w}
\State Delete all paths in the network. \label{ln:delete_paths}
\EndFor
\State \textbf{return:} $\obs_{0:T}, w_{0:T}$ \label{ln:return}
\end{algorithmic}
\end{algorithm}

We consider a network analogous to the current Tor network with 6037 relays as of June $23^{\mathit{rd}}$ 2020. The relays are distributed as follows: $2351$ are guard relays, $2576$ are middle relays, and $1110$ are exit relays (this includes any relays that have both the \textsf{Exit} and \textsf{Guard} flags set). Lacking a ground truth, we used the measured capacity in the Tor consensus document as the actual capacity of the Tor relays in our simulation. The maximum capacity of all relays was $169\,000$ kb/s, while the total capacity of the guard, middle and exit relays in the network are around $42.6 \times 10^6$, $6.7 \times 10^6$ and $17.7 \times 10^6$ kb/s respectively. Hence, the capacity set is $\capacityset = [0, 169000]$. The capacity distributions of the guard, middle, and exit relays are shown in Figure~\ref{fig:100_network_capacities}.

The max-min bandwidth allocation algorithm~\cite{tightrope, msthesis}, in Python assumes that clients will use the full bandwidth allocated to them. Thus we simulated two types of networks: (1) a full utilization network using the max-min bandwidth allocation algorithm where each client path consists of three relays and (2) an under loaded network scenario taking into account client side bottlenecks. To simulate this idea, we adjusted the simulation to add a bandwidth cap to each client flow. This bandwidth cap is enforced by a fourth relay added to each client flow, with a capacity  selected uniformly at random from the interval [8,18]\,KB/s. This interval was chosen based on the full utilization scenario bandwidth distribution. Since the average bandwidth of flows in the full utilization scenario was approximately 22\,KB/s, the cap means that the flows can utilize at most about 60\% of the Tor network capacity.

To use $\dualmethod$, we need to partition $\capacityset$ into bins.
From Figure~\ref{fig:100_network_capacities}, we use the same technique used in ~\cite{mleflowp}. 
We choose the bins to be intervals of the form $[a^{b-1},a^b]$ where $a$ is a strictly positive real number and $b \in [1,...,b_{max}]$ where $b_{max} = \lceil\frac{\log(\max(C[j]))}{\log(a)}\rceil$ for $j \in [n]$. 

\begin{figure}
    \centering
    \includegraphics[width=\columnwidth]{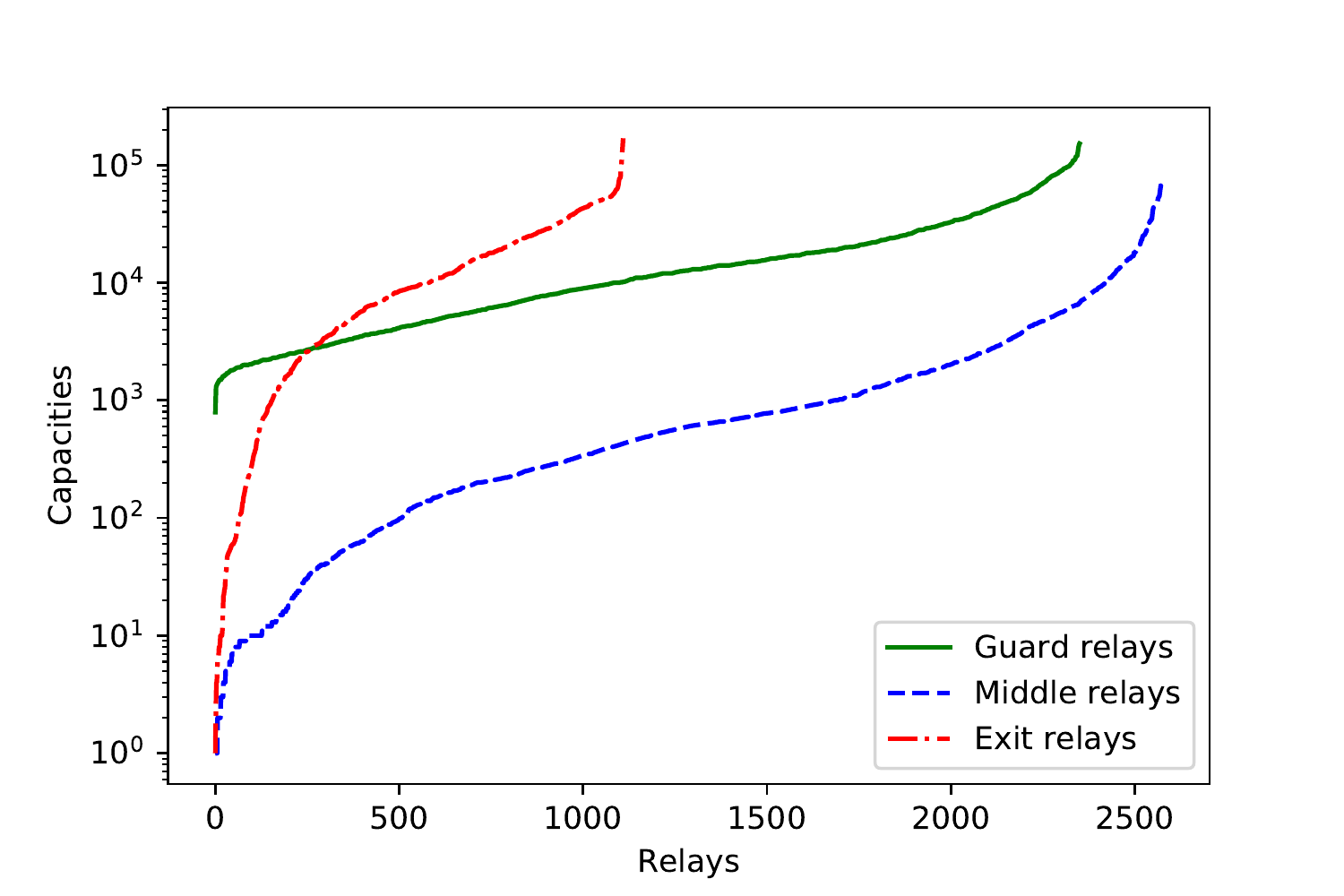}
    \caption{relays capacity distribution.}
    \label{fig:100_network_capacities}
\end{figure}

 \subsection{Low-fidelity sims results and analysis}

 In all of the simulation runs, we chose the number of periods $t$ to be $50$, the initial weight vector $w_0$ to be uniform, the rate of arrival $\lambda_s$ to be $10^{6}$. 
 
  \paragraph{$\dualmethodwhole$ performs better than $\dualmethodone$} While both algorithms of $\dualmethod$ outperformed the other algorithms, $\dualmethodwhole$ had a lower average error than $\dualmethodone$ for all classes of relays. The results are shown in Figure~\ref{fig:diprober_comp}.
 
  \begin{figure}
    \centering
    \includegraphics[width=\columnwidth]{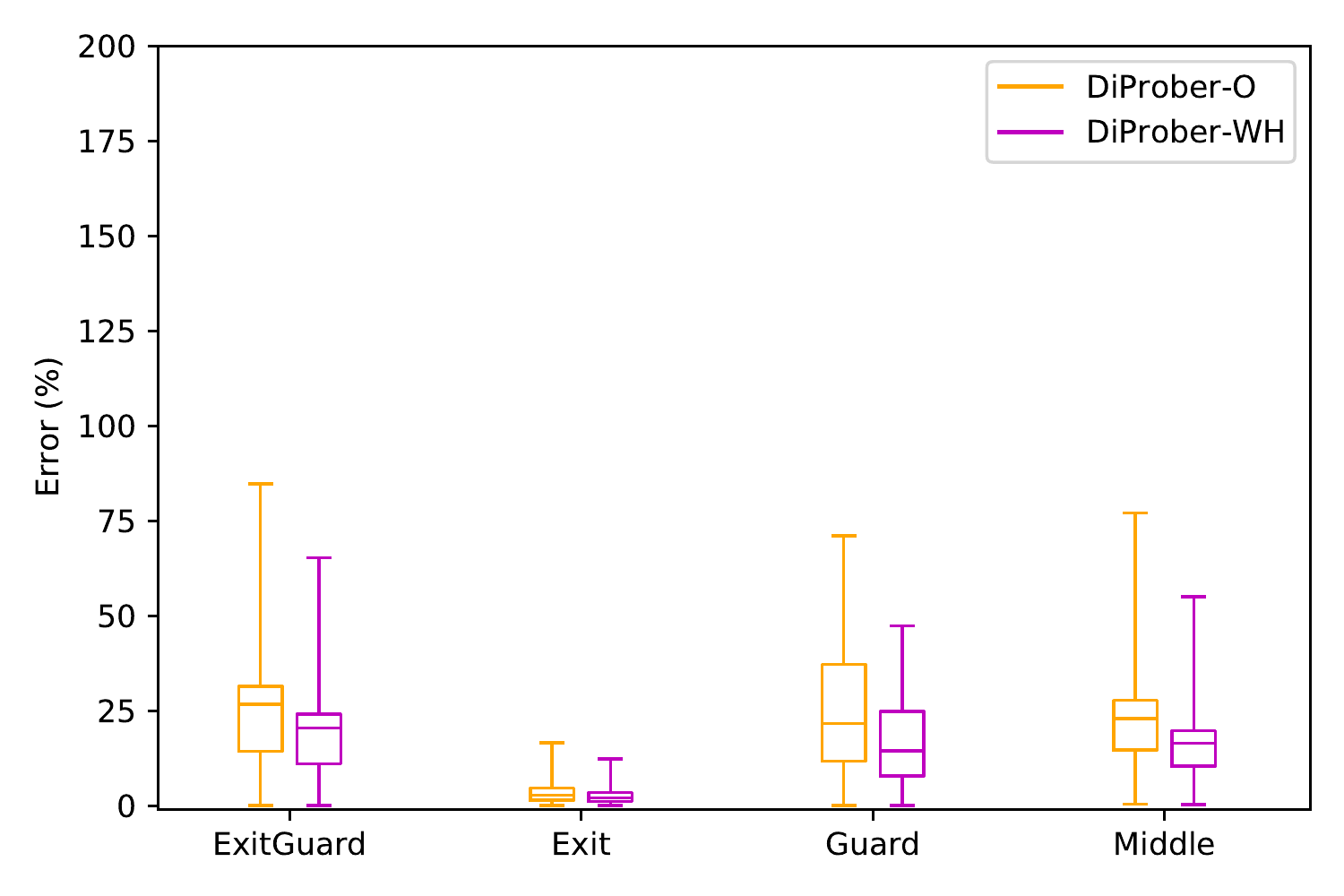}
    \caption{Estimation error distribution after the  $50^{th}$ measurement period in a fully utilized network for both $\dualmethod$ algorithms.}
    \label{fig:diprober_comp}
\end{figure} 
 
 \paragraph{$\dualmethod$ leads to better guard and middle relays estimates in fully loaded networks} We tested the different estimation algorithms on three-relays paths networks. Both $\dualmethod$ algorithms and $\ourmethod$ had lower average error than $\proptorflow$ and $\sbws$. The results are shown in Figure~\ref{fig:error_one_relay}. While the errors for exit relays when using $\dualmethod$ and $\ourmethod$ were close, $\dualmethod$ leads to lower average error for guard and middle relays. The average estimation errors for $\dualmethod$ and $\ourmethod$ stayed below $10\%$ for exit relays, while it was higher for $\proptorflow$  with $72\%$. For guard and middle relays, the error was lowest for  $\dualmethodwhole$ at $18\%$ and $14\%$ respectively. It was higher for $\ourmethod$ at around $25\%$ error and even higher for $\proptorflow$ at $75\%$. 
 
 \begin{figure}
    \centering
    \includegraphics[width=\columnwidth]{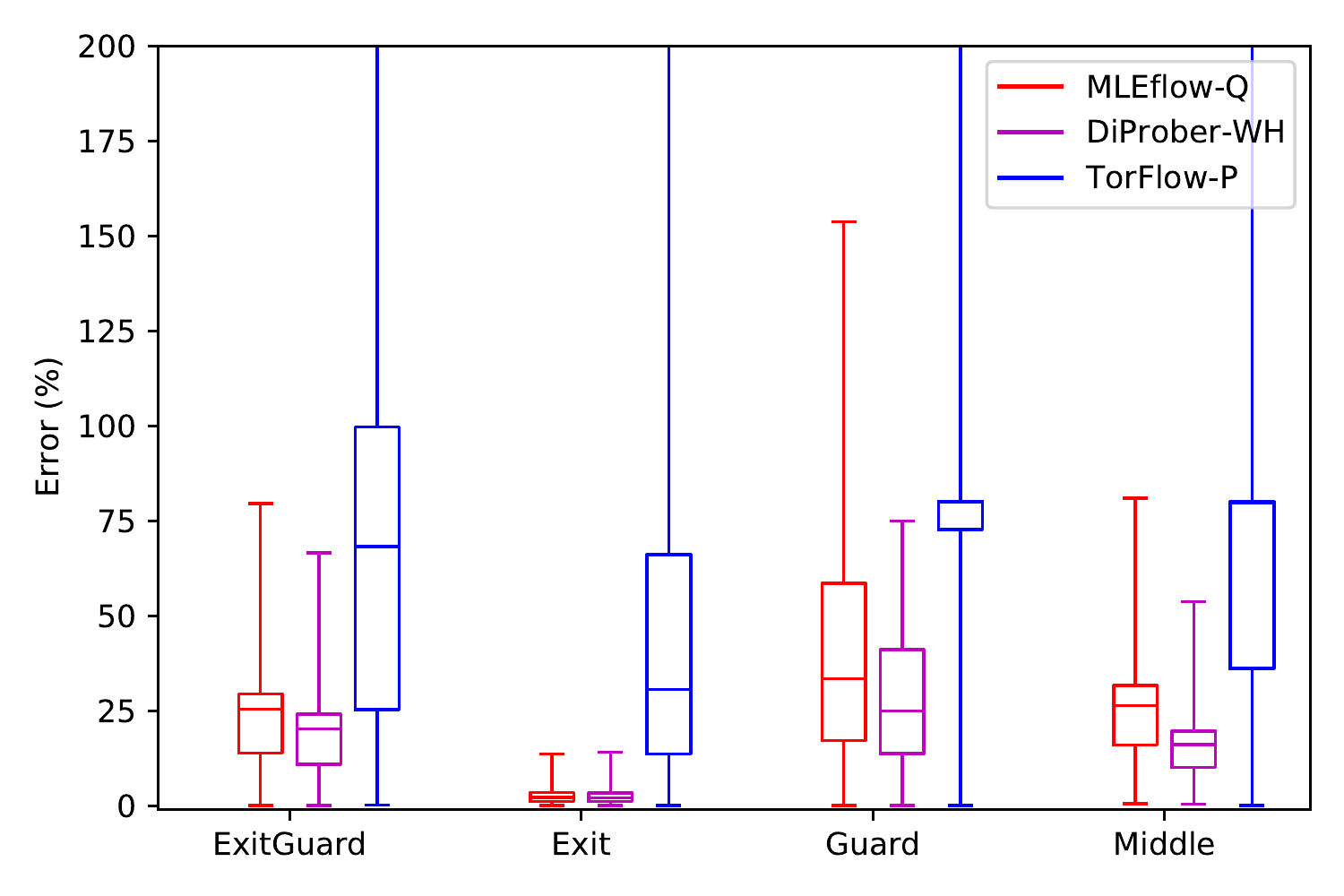}
    \caption{Estimation error distribution after the  $50^{th}$ measurement period in a fully utilized network.}
    \label{fig:error_one_relay}
\end{figure}

 \paragraph{The dual probing algorithm advantage extend to the under-loaded network scenario:} We simulate the case of under-loaded networks as was done in ~\cite{mleflowp} and as described above by adding a forth relay to each path created. $\dualmethod$ outperforms all other algorithms. The error for all classes of relays was below $25\%$ for $\dualmethod$. While the error for exit relays remained relatively low when using $\ourmethod$, other classes estimates all had an average error above $45\%$. The results are shown in Figure~\ref{fig:underloaded}.
 
  \begin{figure}
    \centering
    \includegraphics[width=\columnwidth]{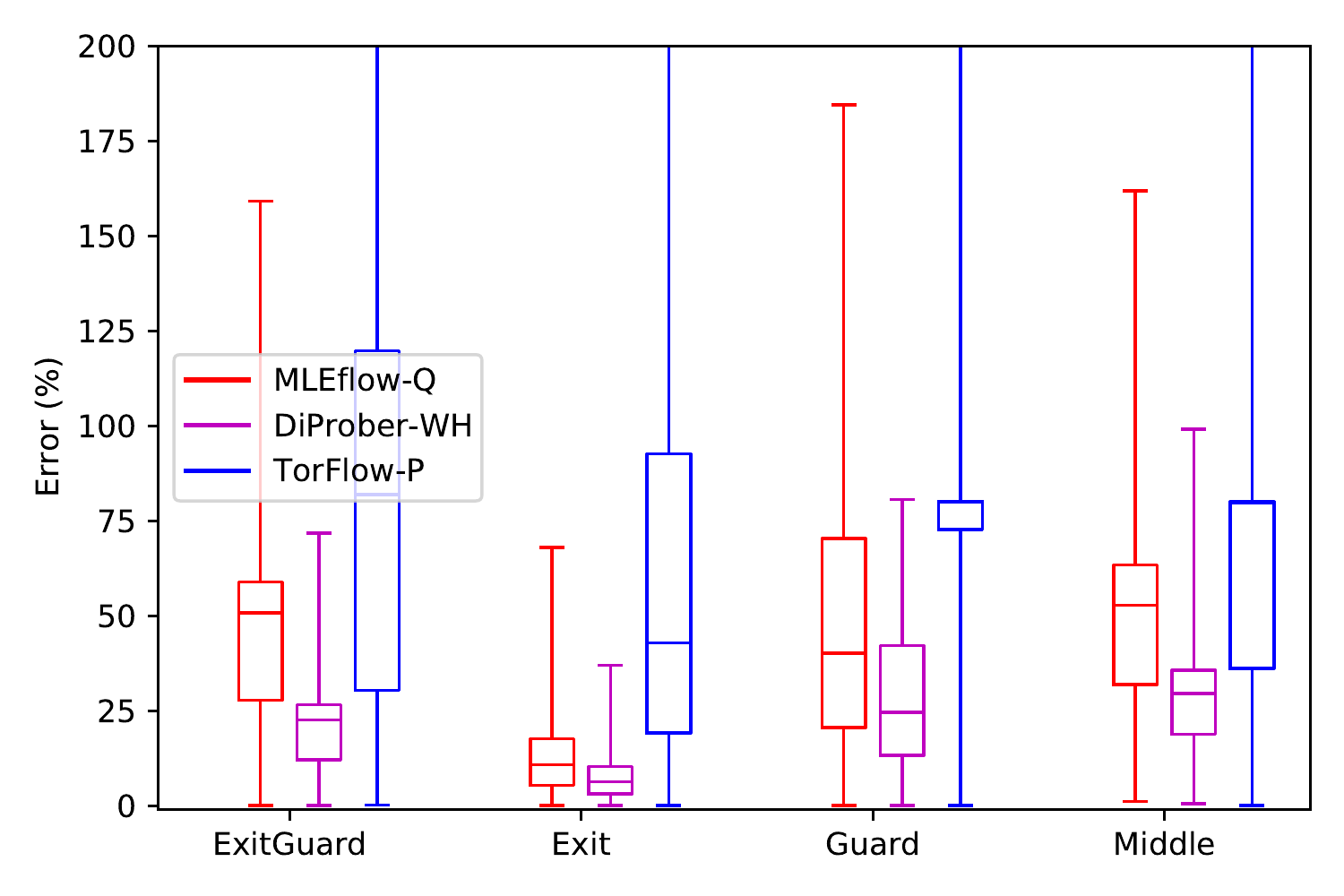}
    \caption{Estimation error distribution after the  $50^{th}$ measurement period in an under-loaded network.}
    \label{fig:underloaded}
\end{figure} 
 
 \paragraph{$\dualmethodwhole$ and $\dualmethodone$ give better and fairer bandwidth allocation than than the other algorithms}
The means of the bandwidths allocated for paths using $\ourmethod$ and $\dualmethod$ are equal to that of the {\em Actual} scenario, while that of $\proptorflow$ is slightly smaller. The advantage of both those methods is the stability of their estimates. The standard deviation of $\dualmethod$ is a maximum of 1.4, while $\ourmethod$ was 1.7, when that of {\em Actual} is around 0.6 and $\proptorflow$ is around 165, orders of magnitude larger.  Moreover, the maximum and minimum bandwidths allocated of our methods are similar to that of {\em Actual} while $\proptorflow$ had orders of magnitude larger maximum. That means that our methods distribute bandwidths more fairly than $\proptorflow$.

\section{Related Work}

Improving the performance of the Tor network has been the subject of much research; we refer the reader to the survey by AlSabah and Goldberg for an overview~\cite{tor-survey}. Here we summarize related work specifically focusing on relay capacity estimation.

Snader and Borisov proposed using \emph{opportunistic measurements}, where each relay measures the bandwidth of each other relay it communicates with as part of normal operation, and designed EigenSpeed~\cite{eigenspeed}, which combines these measurements using principal component analysis to derive a single relay capacity. EigenSpeed was designed to avoid certain types of collusion and misreporting attacks; however, Johnson et al.~\cite{peerflow} discovered that it is subject to a number of other attacks that allow colluding adversaries to inflate their bandwidth. They also designed PeerFlow, which is a more robust mechanism to combine opportunistic measurements from relays with provable limits on inflation attacks. These bounds, however, depend on having a fraction of bandwidth being on trusted nodes, and it has slow convergence properties due to its limitations on changing bandwidth values. 

FlashFlow~\cite{flashflow} is a new proposal to replace TorFlow. FlashFlow uses several servers that measure a relay simultaneously, generating a large network load intended to max out its capacity. FlashFlow has a guaranteed inflation bound of only 33\% but it is based on the assumption that a relay capacity is based on a hard limit that cannot be exceeded, as TorFlow uses traffic that is explicitly labeled for for bandwidth probing. 
In practice, it is often easier and cheaper to obtain high peak bandwidth capability than sustaining the same bandwidth continuously. 

At the same time, TorFlow probes, though not explicitly labeled, can also be identified by their distinct characteristics, and can be used to preferentially forward probe traffic to inflate bandwidth estimates, or to perform sophisticated denial-of-service attacks~\cite{jansen2019point}. A more stealthy approach for bandwidth measurement probing remains an open research question. The simplest and most effective attack on TorFlow, however, is to inflate the observed bandwidth published by the relay~\cite{peerflow}; this attack does not apply to \dualmethod as it does not use the observed bandwidth.
\section{Conclusion}

We have developed a new method for estimating the relay capacities in the Tor network, \dualmethod. We show that \ourmethod fails to accurately estimate relays capacities in under-loaded networks. Our detailed mathematical analysis showed that \dualmethod capacity estimates converge to their true value, while the estimate variance converges to 0, as the number of observations grows.  We validated the performance of \ourmethod with extensive simulations using a our custom flow-based simulator. Our results show that \dualmethod produces much more accurate estimates of relay capacities, which in turn results in much better load balancing of user traffic across the network, as compared with current methods.






\bibliographystyle{IEEEtranS}
\bibliography{refs,nikita}
%

\appendix
\label{sec:appendix}

\subsection{Proofs}
\label{sec:appendix_proofs}


\mledual*
\begin{proof}
As discussed, at each epoch, a relay can fall into the two cases discussed in section~\ref{sec:meas_mech}. We start with \textbf{Case 1}, where the second measurement is equal to double the first measurement of a relay.
When evaluating the objective function of the MLE, the observation random variable of the $j^{\mathit{th}}$ relay $\Obs_i[j]$ can be written as a function of $\Cq $, as if we are assuming $\Cq  = C^*[j]$, and the random variable $X_i[j]$ for $i \in [t]$:
\begin{equation}
    \Obs^2_i[j]=\frac{\Cq - X_i[j] C_{Client}^{avg}}{2}.
    \label{eq:Mal2}
\end{equation}
Recall that we assume that the random variable $X_i[j]$ follows a Poisson distribution with parameter $\lambda_sw_i[j]$ and all users leave at the end of each epoch. Hence, given  $w_i[j]$ for $j \in [n]$, the $\Obs_i[j]$'s at different iterations are independent random variables. Thus \cref{eq:Mal2} can be written as the product of the probability of the independent random variables $[\Obs_1[j],...,\Obs_t[j]]$:
\begin{equation}
\resizebox{0.5\textwidth}{!}{$ 
\begin{multlined}
    C_{t+1}^H[j] = \underset{\Cq  \in \capacityset}{\operatorname{argmax}} \hspace{0.1cm}\Pr_{\pathnumrand_{[t]}[j] \sim  \pois( \rate W_{[t]}[j])}(\Obs^2_{[t]}[j]=\obs^2_{[t]}[j]\hspace{0.1cm}|\nonumber
     W_{[t]}[j] = \weight_{[t]}[j])\\
    C^H_{t+1}[j] = \underset{\Cq  \in \capacityset}{\operatorname{argmax}} \hspace{0.1cm}\prod_{i=0}^{t}\Pr(\Obs^2_i[j]=\obs^2_i[j]\hspace{0.1cm}|\hspace{0.1cm}W_{i}[j]=w_i[j]).\\
    \label{eq:Mal3}
    \end{multlined}$}
\end{equation}
Rearranging \cref{eq:Mal2} results in:
\begin{equation}
    X_i[j]= \frac{\Cq - 2\Obs^2_i[j] }{C_{Client}^{avg}}.
    \label{Mal4}
\end{equation}
When the measurement is made and the observation is fixed, i.e. $\Obs^2_i[j]=\obs^2_i[j]$, the probability in \cref{eq:Mal3} can be expressed in terms of the random variable $X_i[j]$: $x_i[j]=\frac{\Cq - 2\obs^2_i[j] }{C_{Client}^{avg}}$.
\begin{equation}
    C^H_{t+1}[j] = \underset{\Cq  \in \capacityset}{\operatorname{argmax}} \hspace{0.1cm}\prod_{i=0}^{t}\Pr(X_i[j]=x_i[j]|\hspace{0.1cm}W_{i}[j]=w_i[j])
    \label{Mal5}
\end{equation}
Using the Poisson distribution probability mass function, we can write:
\begin{equation}
\begin{multlined}
    C^H_{t+1}[j] = \underset{\Cq  \in \capacityset}{\operatorname{argmax}} \hspace{0.1cm}\prod_{i=0}^{t} e^{-\lambda_sw_i[j]}\frac{1}{(x_i[j])!}(\lambda_sw_i[j])^{x_i[j]}\\
    = \underset{\Cq  \in \capacityset}{\operatorname{argmax}} \hspace{0.1cm}\prod_{i=0}^{t} e^{-\lambda_sw_i[j]}\frac{1}{\big(\frac{\Cq - 2\obs^2_i[j] }{C_{Client}^{avg}}\big)!}(\lambda_sw_i[j])^{\frac{\Cq - 2\obs^2_i[j] }{C_{Client}^{avg}}}\\
    \end{multlined}
    \label{eq:Mal6}
\end{equation}


Considering \textbf{Case 2}, the objective function of the MLE, the observation random variable of the $j^{\mathit{th}}$ relay $\Obs_i[j]$ can be written as a function of $\Cq $, as if we are assuming $\Cq  = C^*[j]$, and the random variable $X_i[j]$ for $i \in [t]$:
\begin{equation}
    \Obs^2_i[j]=\frac{\Cq }{X_i[j]+2}.
    \label{eq:Maln}
\end{equation}
Recall that we assume that the random variable $X_i[j]$ follows a Poisson distribution with parameter $\lambda_sw_i[j]$ and all users leave at the end of each epoch. Hence, given  $w_i[j]$ for $j \in [n]$, the $\Obs^2_i[j]$'s at different iterations are independent random variables. Thus \cref{eq:Maln} can be written as the product of the probability of the independent random variables $[\Obs^2_1[j],...,\Obs^2_t[j]]$:
\begin{equation}
\resizebox{0.5\textwidth}{!}{$ 
\begin{multlined}
    C_{t+1}^H[j] = \underset{\Cq  \in \capacityset}{\operatorname{argmax}} \hspace{0.1cm}\Pr_{\pathnumrand_{[t]}[j] \sim  \pois( \rate W_{[t]}[j])}(\Obs_{[t]}[j]=\obs_{[t]}[j]\hspace{0.1cm}|\nonumber
     W_{[t]}[j] = \weight_{[t]}[j])\\
    C^H_{t+1}[j] = \underset{\Cq  \in \capacityset}{\operatorname{argmax}} \hspace{0.1cm}\prod_{i=0}^{t}\Pr(\Obs_i[j]=\obs_i[j]\hspace{0.1cm}|\hspace{0.1cm}W_{i}[j]=w_i[j]).\\
    \label{eq:Maln3}
    \end{multlined}$}
\end{equation}
Rearranging \cref{eq:Maln} results in:
\begin{equation}
    X_i[j]= \frac{\Cq }{\Obs^2_i[j]}-2.
    \label{Maln4}
\end{equation}
When the measurement is made and the observation is fixed, i.e. $\Obs^2_i[j]=\obs_i[j]$, the probability in \cref{eq:Maln3} can be expressed in terms of the random variable $X_i[j]$: $X_i[j]=\frac{\Cq }{\obs^2_i[j]}-2$.
\begin{equation}
    C^H_{t+1}[j] = \underset{\Cq  \in \capacityset}{\operatorname{argmax}} \hspace{0.1cm}\prod_{i=0}^{t}\Pr(X_i[j]=x_i[j]|\hspace{0.1cm}W_{i}[j]=w_i[j])
    \label{Maln5}
\end{equation}
Using the Poisson distribution probability mass function, we can write:
\begin{equation}
\begin{multlined}
    C^H_{t+1}[j] = \underset{\Cq  \in \capacityset}{\operatorname{argmax}} \hspace{0.1cm}\prod_{i=0}^{t} e^{-\lambda_sw_i[j]}\frac{1}{(x_i[j])!}(\lambda_sw_i[j])^{x_i[j]}\\
    = \underset{\Cq  \in \capacityset}{\operatorname{argmax}} \hspace{0.1cm}\prod_{i=0}^{t} e^{-\lambda_sw_i[j]}\frac{1}{\big(\frac{\Cq }{\obs^2_i[j]}-2\big)!}(\lambda_sw_i[j])^{\frac{\Cq }{\obs^2_i[j]}-2}\\
    \end{multlined}
    \label{eq:Maln6}
\end{equation}

\end{proof}

\oneobsclosedform*
\begin{proof}
\underline{\textbf{Case 1:}} As discussed previously, we know that we are in this case if $\obs_i^1[j] = 2 \obs_i^2[j]$. We denote $C_{avg,i}$ the average bandwidth used by each client in the network during the $i^{th}$ epoch. Hence since the relay is not bottlenecked, we can say that $\obs_i^1[j]=C^*[j]-\pathnum_i[j]C_{avg,i}$; hence, given $\weight_i$, the measurement of the $j^{\mathit{th}}$ relay at the $i^{\mathit{th}}$ iteration is a random variable $\obsrand_i^1[j] = C^*[j]-\pathnumrand_i[j]C_{avg,i}$.

W will use superscript $D1$ to indicate that we are only using the most recent value of $\obs_t[j]$ and $\weight_t[j]$. Using maximum likelihood estimation, we have
\begin{small}
\begin{align}
\label{eq:mle_case1_onemeas_def}
&C_{t+1}^R[j] = \underset{\Cq \in \capacityset}{\operatorname{argmax}}\ f(\Cq, \obs_{t}[j], \weight_{t}[j]), \text{ where} \\
 \label{eq:Mal}
&f(\Cq, \obs_{t}[j], \weight_{t}[j])= \hspace{-5mm} \Pr_{\pathnumrand_{t}[j] \sim  \pois ( \rate w_{t}[j])} \left( \Cq-\pathnumrand_t[j]C_{avg,t} = \obs_i^1[j] \right),
\end{align}
\end{small}

\begin{equation*}
\resizebox{0.5\textwidth}{!}{$
\begin{multlined}
        C_{t+1}^{D1}[j]
        = \underset{\Cq \in \capacityset}{\operatorname{argmax}}\  \Pr_{\pathnumrand_{t}[j] \sim  \pois ( \rate w_{t}[j])} \left( \Cq-\pathnumrand_t[j]C_{avg,t} = \obs_t^1[j] \right)\\
        = \underset{\Cq \in \capacityset}{\operatorname{argmax}}\  \Pr_{\pathnumrand_{t}[j] \sim  \pois ( \rate w_{t}[j])} \left( \pathnumrand_t[j] =\frac{\Cq-\obs_t^1[j]}{C_{avg,t}} \right)\\
        = \underset{\Cq \in \capacityset}{\operatorname{argmax}}\ \log \left( \Pr_{\pathnumrand_{t}[j] \sim  \pois ( \rate w_{t}[j])} \left( \pathnumrand_t[j] =\frac{\Cq-\obs_t^1[j]}{C_{avg,t}} \right) \right)\\
        = \underset{\Cq \in \capacityset}{\operatorname{argmax}}\ \log \left( \exp \left( - \rate w_{t}[j] \right)\frac{1}{\left(\frac{\Cq-\obs_t^1[j]}{C_{avg,t}}\right)!}\left( \rate w_{t}[j]\right)^{\frac{\Cq-\obs_t^1[j]}{C_{avg,t}}} \right)\\
         = \underset{\Cq \in \capacityset}{\operatorname{argmax}}\ - \rate w_{t}[j] - \log \left(\left(\frac{\Cq-\obs_t^1[j]}{C_{avg,t}}\right)!\right)+\left( \frac{\Cq-\obs_t^1[j]}{C_{avg,t}}\right) \log \left(\rate w_{t}[j] \right) \\
    \end{multlined}$}
\end{equation*}
Using Stirling approximation for the second term, we get

\begin{equation*}
\resizebox{0.5\textwidth}{!}{$
\begin{multlined} 
        C_{t+1}^{D1}[j]
         = \underset{\Cq \in \capacityset}{\operatorname{argmax}}\ - \rate w_{t}[j] - \left(\frac{\Cq-\obs_t^1[j]}{C_{avg,t}}\right) \log \left(\frac{\Cq-\obs_t^1[j]}{C_{avg,t}}\right)\\+\left( \frac{\Cq-\obs_t^1[j]}{C_{avg,t}}\right)+\left( \frac{\Cq-\obs_t^1[j]}{C_{avg,t}}\right) \log \left(\rate w_{t}[j] \right) \\
    \end{multlined}$}
\end{equation*}

In order to find maximum, we differentiate the right side with respect to $\Cq$ and equate it to zero:

\begin{equation*}
\resizebox{0.5\textwidth}{!}{$
\begin{multlined}
-\frac{1}{C_{avg,t}}\log \left(\frac{\Cq-\obs_t^1[j]}{C_{avg,t}}\right)-  \left(\frac{\Cq-\obs_t^1[j]}{C_{avg,t}}\right)\frac{\frac{1}{C_{avg,t}}}{\left(\frac{\Cq-\obs_t^1[j]}{C_{avg,t}}\right)}+\frac{1}{C_{avg,t}}+ \frac{1}{C_{avg,t}} \log \left(\rate w_{t}[j] \right) = 0\\
-\frac{1}{C_{avg,t}}\log \left(\frac{\Cq-\obs_t^1[j]}{C_{avg,t}}\right)-  \frac{1}{C_{avg,t}}+\frac{1}{C_{avg,t}}+ \frac{1}{C_{avg,t}} \log \left(\rate w_{t}[j] \right) = 0\\
-\log \left(\frac{\Cq-\obs_t^1[j]}{C_{avg,t}}\right)+  \log \left(\rate w_{t}[j] \right) = 0\\
\log \left(\frac{C_{avg,t}\rate w_{t}[j]}{\Cq-\obs_t^1[j]} \right) = 0\\
\frac{C_{avg,t}\rate w_{t}[j]}{\Cq-\obs_t^1[j]}  = 1\\
    \end{multlined}$}
\end{equation*}
Thus $\Cq= \underbrace{\rate w_{t}[j]C_{avg,t}}_\text{expected bandwidth used by clients} + \underbrace{\obs_t^1[j]}_\text{bandwidth left unused}$\\

\underline{\textbf{Case 2:}} We know that we are in this case if $\obs_i^1[j] \neq 2 \obs_i^2[j]$.  Since the relay is bottlenecked when we add the second probe in all cases described previously, we can say that $\obs_i^2[j]=\frac{C^*[j]}{\pathnum_i[j]+2}$; hence, given $\weight_i$, the measurement of the $j^{\mathit{th}}$ relay at the $i^{\mathit{th}}$ iteration is a random variable $\obsrand_i^2[j] = \frac{C^*[j]}{\pathnumrand_i[j]+2}$.

Using maximum likelihood estimation, we have
\begin{small}
\begin{align}
\label{eq:mle_case1_onemeas_def}
&C_{t+1}^{D1}[j] = \underset{\Cq \in \capacityset}{\operatorname{argmax}}\ f(\Cq, \obs_{t}[j], \weight_{t}[j]), \text{ where} \\
 \label{eq:Mal}
&f(\Cq, \obs_{t}[j], \weight_{t}[j])= \hspace{-5mm} \Pr_{\pathnumrand_{t}[j] \sim  \pois ( \rate w_{t}[j])} \left( \frac{\Cq}{\pathnumrand_i[j]+2} = \obs_i^2[j] \right),
\end{align}
\end{small}

\begin{equation*}
\resizebox{0.5\textwidth}{!}{$
\begin{multlined}
        C_{t+1}^{D1}[j]
        = \underset{\Cq \in \capacityset}{\operatorname{argmax}}\  \Pr_{\pathnumrand_{t}[j] \sim  \pois ( \rate w_{t}[j])} \left( \frac{\Cq}{\pathnumrand_i[j]+2} = \obs_t^2[j] \right)\\
        = \underset{\Cq \in \capacityset}{\operatorname{argmax}}\  \Pr_{\pathnumrand_{t}[j] \sim  \pois ( \rate w_{t}[j])} \left( \pathnumrand_t[j] =\frac{\Cq}{\obs_t^2[j]}-2 \right)\\
        = \underset{\Cq \in \capacityset}{\operatorname{argmax}}\ \log \left( \Pr_{\pathnumrand_{t}[j] \sim  \pois ( \rate w_{t}[j])} \left( \pathnumrand_t[j] =\frac{\Cq}{\obs_t^2[j]}-2 \right) \right)\\
        = \underset{\Cq \in \capacityset}{\operatorname{argmax}}\ \log \left( \exp \left( - \rate w_{t}[j] \right)\frac{1}{\left(\frac{\Cq}{\obs_t^2[j]}-2\right)!}\left( \rate w_{t}[j]\right)^{\frac{\Cq}{\obs_t^2[j]}-2} \right)\\
         = \underset{\Cq \in \capacityset}{\operatorname{argmax}}\ - \rate w_{t}[j] - \log \left(\left(\frac{\Cq}{\obs_t^2[j]}-2\right)!\right)+\left( \frac{\Cq}{\obs_t^2[j]}-2\right) \log \left(\rate w_{t}[j] \right) \\
    \end{multlined}$}
\end{equation*}
Using Stirling approximation for the second term, we get

\begin{equation*}
\resizebox{0.5\textwidth}{!}{$
\begin{multlined}
        C_{t+1}^R[j]
         = \underset{\Cq \in \capacityset}{\operatorname{argmax}}\ - \rate w_{t}[j] - \left(\frac{\Cq}{\obs_t^2[j]}-2\right) \log \left(\frac{\Cq}{\obs_t^2[j]}-2\right)+\left( \frac{\Cq}{\obs_t^2[j]}-2\right)+\left( \frac{\Cq}{\obs_t^2[j]}-2\right) \log \left(\rate w_{t}[j] \right) \\
    \end{multlined}$}
\end{equation*}

In order to find the optimum, we differentiate the right side with respect to $\Cq$ and equate it to zero:

\begin{equation*}
\resizebox{0.5\textwidth}{!}{$
\begin{multlined}
-\frac{1}{\obs_t^2[j]}\log \left(\frac{\Cq}{\obs_t^2[j]}-2\right)-  \left(\frac{\Cq}{\obs_t^2[j]}-2\right)\frac{\frac{1}{\obs^2_t[j]}}{\left(\frac{\Cq}{\obs_t^2[j]}-2\right)}+\frac{1}{\obs_t^2[j]}+ \frac{1}{\obs_t^2[j]} \log \left(\rate w_{t}[j] \right) = 0\\
-\frac{1}{\obs_t^2[j]}\log \left(\frac{\Cq}{\obs_t^2[j]}-2\right)-  \frac{1}{\obs_t^2[j]}+\frac{1}{\obs_t^2[j]}+ \frac{1}{\obs_t^2[j]} \log \left(\rate w_{t}[j] \right) = 0\\
-\log \left(\frac{\Cq}{\obs_t^2[j]}-2\right)+  \log \left(\rate w_{t}[j] \right) = 0\\
\log \left(\frac{\rate w_{t}[j]\obs_t^2[j]}{\Cq-2\obs_t^2[j]} \right) = 0\\
\frac{\rate w_{t}[j]\obs_t^2[j]}{\Cq-2\obs_t^2[j]}  = 1\\
    \end{multlined}$}
\end{equation*}
Thus $\Cq= \obs_t^2[j]\left(\underbrace{\rate w_{t}[j]}_\text{expected number of users} + 2 \right)$\\
\end{proof}

\convguarantees*
\begin{proof}
We start by considering a relay that falls into \textbf{case 1} for its whole operation and hence is never bottlenecked. From Theorem~\ref{thm:mledual}, we can write:

\begin{equation*}
\resizebox{0.5\textwidth}{!}{$
\begin{multlined}
        C_{t+1}^D[j]
        = \underset{\Cq \in \capacityset}{\operatorname{argmax}}\  \prod_{i=0}^{t} \Pr_{\pathnumrand_{i}[j] \sim  \pois ( \rate w_{i}[j])} \left( \Cq-\pathnumrand_i[j]C_{avg} = \obs_i^1[j] \right)\\
        = \underset{\Cq \in \capacityset}{\operatorname{argmax}}\ \prod_{i=0}^{t} \Pr_{\pathnumrand_{i}[j] \sim  \pois ( \rate w_{i}[j])} \left( \pathnumrand_i[j] =\frac{\Cq-\obs_i^1[j]}{C_{avg}} \right)\\
        = \underset{\Cq \in \capacityset}{\operatorname{argmax}}\ \log \left(\prod_{i=0}^{t} \Pr_{\pathnumrand_{i}[j] \sim  \pois ( \rate w_{i}[j])} \left( \pathnumrand_i[j] =\frac{\Cq-\obs_i^1[j]}{C_{avg}} \right) \right)\\
        = \underset{\Cq \in \capacityset}{\operatorname{argmax}}\ \log \left( \prod_{i=0}^{t} \exp \left( - \rate w_{i}[j] \right)\frac{1}{\left(\frac{\Cq-\obs_i^1[j]}{C_{avg}}\right)!}\left( \rate w_{i}[j]\right)^{\frac{\Cq-\obs_i^1[j]}{C_{avg}}} \right)\\
         = \underset{\Cq \in \capacityset}{\operatorname{argmax}}\sum_{i=0}^{t} - \rate w_{i}[j] - \log \left(\left(\frac{\Cq-\obs_i^1[j]}{C_{avg}}\right)!\right)+\left( \frac{\Cq-\obs_i^1[j]}{C_{avg}}\right) \log \left(\rate w_{i}[j] \right) \\
         = \underset{\Cq \in \capacityset}{\operatorname{argmax}}\sum_{i=0}^t - \rate w_{i}[j] - \left(\frac{\Cq-\obs_i^1[j]}{C_{avg}}\right) \log \left(\frac{\Cq-\obs_i^1[j]}{C_{avg}}\right)+\left( \frac{\Cq-\obs_i^1[j]}{C_{avg}}\right)+\left( \frac{\Cq-\obs_i^1[j]}{C_{avg}}\right) \log \left(\rate w_{i}[j] \right) \\
    \end{multlined}$}
\end{equation*}

We differentiate the right side with respect to $\Cq$ and equate it to zero:

\begin{equation*}
\resizebox{0.5\textwidth}{!}{$
\begin{multlined}
\sum_{i=0}^t -\frac{1}{C_{avg}}\log \left(\frac{\Cq-\obs_i^1[j]}{C_{avg}}\right)-  \left(\frac{\Cq-\obs_i^1[j]}{C_{avg}}\right)\frac{\frac{1}{C_{avg}}}{\left(\frac{\Cq-\obs_i^1[j]}{C_{avg}}\right)}+\frac{1}{C_{avg}}+ \frac{1}{C_{avg}} \log \left(\rate w_{i}[j] \right) = 0\\
\sum_{i=0}^t-\frac{1}{C_{avg}}\log \left(\frac{\Cq-\obs_i^1[j]}{C_{avg}}\right)-  \frac{1}{C_{avg}}+\frac{1}{C_{avg}}+ \frac{1}{C_{avg}} \log \left(\rate w_{i}[j] \right) = 0\\
\sum_{i=0}^t -\log \left(\frac{\Cq-\obs_i^1[j]}{C_{avg}}\right)+  \log \left(\rate w_{i}[j] \right) = 0\\
    \end{multlined}$}
\end{equation*}

Using linearization to solve for the optimum $\Cq$, we have that

\begin{equation*}
\resizebox{0.5\textwidth}{!}{$
\begin{multlined}
f'_t(\Cq)=\sum_{i=0}^t -\log \left(\Cq-\obs_i^1[j]\right)+  \sum_{i=0}^t\log \left(\rate w_{i}[j]C_{avg} \right)\\ = f_{t-1}'(\Cq) +\log \left(\rate w_{t}[j]C_{avg} \right) - \log \left(\Cq-\obs_t^1[j]\right)\\ 
    \end{multlined}$}
\end{equation*}
We know that $f'_{t}(\Cq_{t})=0$ and $f'_{t-1}(\Cq_{t-1})=0$ thus

\begin{equation*}
\resizebox{0.5\textwidth}{!}{$
\begin{multlined}
f'_t(\Cq_{t-1})=f_{t-1}'(\Cq_{t-1}) +\log \left(\rate w_{t}[j]C_{avg} \right) - \log \left(\Cq_{t-1}-\obs_t^1[j]\right)\\
= \log \left(\rate w_{t}[j]C_{avg} \right) - \log \left(\Cq_{t-1}-\obs_t^1[j]\right)\\     \end{multlined}$}
\end{equation*}
We also have $f''_t(\Cq)= - \sum_{i=0}^t\frac{1}{\Cq-\obs_i^1[j]}= f''_{t-1}(\Cq)-\frac{1}{\Cq-\obs_t^1[j]}$.

By linearization we have $f'_t(\Cq_t)= f'_t(\Cq_{t-1})+f''_t(\Cq_{t-1})(\Cq_t-\Cq_{t-1})$ and thus,

\begin{equation*}
\begin{multlined}
\Cq_t=\frac{f'_t(\Cq_t)-f'_t(\Cq_{t-1})}{f''_t(\Cq_{t-1})}+\Cq_{t-1}
 \end{multlined}
\end{equation*}
Thus we can find an iterative solution of the optimization,

\begin{equation*}
\resizebox{0.5\textwidth}{!}{$
\begin{multlined}
\Cq_0=\obs_0^1[j]+\rate w_0[j]C_{avg} \hspace{3 mm} \text{and}\\
\Cq_t=\frac{\log \left(\rate w_{t}[j]C_{avg} \right) - \log \left(\Cq_{t-1}-\obs_t^1[j]\right)}{\sum_{i=0}^t\frac{1}{\Cq_{t-1}-\obs_i^1[j]}}+\Cq_{t-1}
\end{multlined}$}
\end{equation*}

Finding the steady state convergence of the above iterative formulation:

\begin{equation}
\begin{multlined}
\Cq=\frac{\log \left(\rate w_{t}[j]C_{avg} \right) - \log \left(\Cq-\obs_t^1[j]\right)}{\sum_{i=0}^t\frac{1}{\Cq-\obs_i^1[j]}}+\Cq\\
\log(\rate w_t[j]C_{avg})= \ log(\Cq-\obs_t^1[j])\\
\rate w_t[j]C_{avg}=\Cq - m_t^1[j] \\
\end{multlined}
\label{eq:steady}
\end{equation}

Hence the expected value of the estimate $E(\Cq_t)= C^*[j]$ since the right hand side of Equation~\ref{eq:steady} is equal to the actual capacity of the relay if it was never bottlenecked on average.

Now we consider \textbf{case 2}, where the relay is always bottlenecked for its whole measurement history. In this proof we will drop the upperscript $2$ from the measurement since we are only dealing with the second measurement of a relay.As we derived in \cref{eq:Maln6}, we know that for any $j \in [n]$, the weight at iteration $(t+1)$ should satisfy the following equation:
\begin{equation}
    C^D_{t+1}[j] = \underset{\Cq  \in \capacityset}{\operatorname{argmax}} \hspace{0.1cm}\prod_{i=0}^{t} e^{-\lambda_sw_i[j]}\frac{1}{\Big(\frac{\Cq }{\obs_i[j]}-2\Big)!}(\lambda_sw_i[j])^{\frac{\Cq }{\obs_i[j]}-2}
    \label{eq:Mal8}
\end{equation}
Since the logarithm function is a strictly increasing function, the maximum likelihood estimate of the capacity of a relay $j \in [n]$ using full history can be found:
\begin{equation}
\resizebox{0.5\textwidth}{!}{$
\begin{multlined}
    C^D_{t+1}[j] = \underset{\Cq  \in \capacityset}{\operatorname{argmax}} \hspace{0.1cm}\prod_{i=0}^{t} e^{-\lambda_sw_i[j]}\frac{1}{\Big(\frac{\Cq }{\obs_i[j]}-2\Big)!}(\lambda_sw_i[j])^{\frac{\Cq }{\obs_i[j]}-2}\\
    =\underset{\Cq  \in \capacityset}{\operatorname{argmax}} \hspace{0.1cm}\prod_{i=0}^{t} e^{-\lambda_sw_i[j]}\frac{(\frac{\Cq }{\obs_i[j]})^2}{\Big(\frac{\Cq }{\obs_i[j]}\Big)!}(\lambda_sw_i[j])^{\frac{\Cq }{\obs_i[j]}}(\lambda_sw_i[j])^{-2}\\
    =\underset{\Cq  \in \capacityset}{\operatorname{argmax}} \hspace{0.1cm}\sum_{i=0}^{t} -\lambda_sw_i[j]+2\log\Big(\frac{\Cq }{\obs_i[j]}\Big)-\log\bigg(\Big(\frac{\Cq }{\obs_i[j]}\Big)!\bigg)+\frac{\Cq }{\obs_i[j]}\log(\lambda_sw_i[j])-2\log(\lambda_sw_i[j])
\end{multlined}$}
    \label{eq:Mal9}
\end{equation}

Using Stirling's approximation, we have $\log(x!)\approx x\log(x)-x$. Thus substituting in \cref{eq:Mal9}:
\begin{equation}
\resizebox{0.5\textwidth}{!}{$
\begin{multlined}
    C^H_{t+1}[j] = \\
    \underset{\Cq  \in \capacityset}{\operatorname{argmax}} \hspace{0.1cm}\sum_{i=0}^{t} -\lambda_sw_i[j]+2\log\Big(\frac{\Cq }{\obs_i[j]}\Big)-\frac{\Cq }{\obs_i[j]}\log\Big(\frac{\Cq }{\obs_i[j]}\Big)+\frac{\Cq }{\obs_i[j]}+\frac{\Cq }{\obs_i[j]}\log(\lambda_sw_i[j])-2\log(\lambda_sw_i[j])
\end{multlined}$}
    \label{Mal10}
\end{equation}
Hence in order to find $C^H_{t+1}[j]$, we differentiate the right hand side of \cref{Mal10} with respect to $\Cq $, and find the value of $C^H_{t+1}[j]$ for which the derivative is zero.

\begin{equation}
\resizebox{0.5\textwidth}{!}{$
\begin{multlined}
   \sum_{i=0}^t\frac{2}{\obs_i[j]}\frac{1}{\frac{\Cq }{\obs_i[j]}}-\frac{1}{\obs_i[j]}\log\Big(\frac{\Cq }{\obs_i[j]}\Big)-\frac{1}{\obs_i[j]}+\frac{1}{\obs_i[j]}+\frac{1}{\obs_i[j]}\log(\lambda_sw_i[j])=0\\
   \sum_{i=0}^t\frac{2}{\Cq }-\frac{1}{\obs_i[j]}\log(\Cq )+\frac{1}{\obs_i[j]}\log(\obs_i[j])+\frac{1}{\obs_i[j]}\log(\lambda_sw_i[j])=0\\
   \sum_{i=0}^t\frac{2}{\Cq }-\frac{1}{\obs_i[j]}\log(\Cq )+\frac{1}{\obs_i[j]}\log(\obs_i[j]\lambda_sw_i[j])=0\\
    \sum_{i=0}^t\frac{1}{\obs_i[j]}\log(\Cq )=\sum_{i=0}^t\frac{2}{\Cq }+\frac{1}{\obs_i[j]}\log(\obs_i[j]\lambda_sw_i[j])\\
    \log(\Cq )\Big(\sum_{i=0}^t\frac{1}{\obs_i[j]}\Big)=\frac{2(t+1)}{\Cq }+\sum_{i=0}^t\frac{1}{\obs_i[j]}\log(\obs_i[j]\lambda_sw_i[j])\\
     \log(\Cq )=\frac{2(t+1)}{\Cq (\sum_{i=0}^t\frac{1}{\obs_i[j]})}+\frac{\sum_{i=0}^t\frac{1}{\obs_i[j]}\log(\obs_i[j]\lambda_sw_i[j])}{\sum_{i=0}^t\frac{1}{\obs_i[j]}}\\
     \log(\Cq )-\frac{2(t+1)}{\Cq (\sum_{i=0}^t\frac{1}{\obs_i[j]})}=\frac{\sum_{i=0}^t\frac{1}{\obs_i[j]}\log(\obs_i[j]\lambda_sw_i[j])}{\sum_{i=0}^t\frac{1}{\obs_i[j]}}\\
     \Cq e^{-\frac{2(t+1)}{\Cq (\sum_{i=0}^t\frac{1}{\obs_i[j]})}}=e^{\frac{\sum_{i=0}^t\frac{1}{\obs_i[j]}\log(\obs_i[j]\lambda_sw_i[j])}{\sum_{i=0}^t\frac{1}{\obs_i[j]}}}\\
     \frac{\sum_{i=0}^t\frac{1}{\obs_i[j]}}{2(t+1)}\Cq  e^{-\frac{2(t+1)}{\Cq (\sum_{i=0}^t\frac{1}{\obs_i[j]})}}=\frac{\sum_{i=0}^t\frac{1}{\obs_i[j]}}{2(t+1)}e^{\frac{\sum_{i=0}^t\frac{1}{\obs_i[j]}\log(\obs_i[j]\lambda_sw_i[j])}{\sum_{i=0}^t\frac{1}{\obs_i[j]}}}\\
\end{multlined}$}
    \label{eq:Mal11}
\end{equation}
Letting $z=\frac{2(t+1)}{\Cq (\sum_{i=0}^t\frac{1}{\obs_i[j]})}$ in \cref{eq:Mal11}, we have:
\begin{equation}
\begin{multlined}
     \frac{1}{z}e^{-z}=\frac{\sum_{i=0}^t\frac{1}{\obs_i[j]}}{2(t+1)}e^{\frac{\sum_{i=0}^t\frac{1}{\obs_i[j]}\log(\obs_i[j]\lambda_sw_i[j])}{\sum_{i=0}^t\frac{1}{\obs_i[j]}}}\\
     \frac{1}{ze^{z}}=\frac{\sum_{i=0}^t\frac{1}{\obs_i[j]}}{2(t+1)}e^{\frac{\sum_{i=0}^t\frac{1}{\obs_i[j]}\log(\obs_i[j]\lambda_sw_i[j])}{\sum_{i=0}^t\frac{1}{\obs_i[j]}}}\\
     ze^z=\frac{2(t+1)}{\sum_{i=0}^t\frac{1}{\obs_i[j]}}e^{-\frac{\sum_{i=0}^t\frac{1}{\obs_i[j]}\log(\obs_i[j]\lambda_sw_i[j])}{\sum_{i=0}^t\frac{1}{\obs_i[j]}}}\\
\end{multlined}
    \label{eq:Mal12}
\end{equation}
We know that the inverse image of the function $ze^z$ is the Lambert W function which has real solutions along its principal branch  for $z>-\frac{1}{e}$, denoted $W_0$. Thus we can solve for $z$:
\begin{equation}
\resizebox{0.5\textwidth}{!}{$
\begin{multlined}
    \frac{2(t+1)}{\Cq (\sum_{i=0}^t\frac{1}{\obs_i[j]})}=z=W_0\Bigg(\frac{2(t+1)}{\sum_{i=0}^t\frac{1}{\obs_i[j]}}e^{-\frac{\sum_{i=0}^t\frac{1}{\obs_i[j]}\log(\obs_i[j]\lambda_sw_i[j])}{\sum_{i=0}^t\frac{1}{\obs_i[j]}}}\Bigg)\\
\end{multlined}$}
\label{eq:Mal13}
\end{equation}

And hence solving for $\Cq $:
\begin{equation}
\resizebox{0.5\textwidth}{!}{$
\begin{multlined}
    C^D_{t+1}[j] = \Cq = \frac{2(t+1)}{\sum_{i=0}^t\frac{1}{\obs_i[j]}}
    \frac{1}{W_0\Bigg(\frac{2(t+1)}{\sum_{i=0}^t\frac{1}{\obs_i[j]}}e^{-\frac{\sum_{i=0}^t\frac{1}{\obs_i[j]}\log(\obs_i[j]\lambda_sw_i[j])}{\sum_{i=0}^t\frac{1}{\obs_i[j]}}}\Bigg)}\\
\end{multlined}$}
    \label{eq:Mal14}
\end{equation}

\begin{equation}
\resizebox{0.5\textwidth}{!}{$
\begin{multlined}
    C^D_{t+1}[j] = \frac{2(t+1)}{\sum_{i=0}^t\frac{1}{\obs_i[j]}}
    \frac{1}{W_0\Bigg(\frac{2(t+1)}{\sum_{i=0}^t\frac{1}{\obs_i[j]}}e^{-\frac{\sum_{i=0}^t\frac{1}{\obs_i[j]}\log(\obs_i[j]\rate w_i[j])}{\sum_{i=0}^t\frac{1}{\obs_i[j]}}}\Bigg)},\\
\end{multlined}$}
    \label{eq:Mal15}
\end{equation}
 where $W_0$ is the Lambert $W$ function along the principal branch.
The Lambert $W$  function is the multi-valued complex function $(ze^{z})^{-1}$ and $W_0$ is the unique-valued real function that takes the unique real value of $W$ when $z > \frac{-1}{e}$. Implementations of Lambert function exist in multiple software libraries~\footnote{\url{https://kite.com/python/docs/mpmath.lambertw}}.

 $W_0$ has the following Taylor series expansion for $z$ in the neighborhood of 0: $W_0(z)=z+o(z^2)$. Moreover, the argument of $W_0$ is small if the rate of users arrival to the network $\rate$ is large enough. Hence, the Taylor expansion around zero is valid and therefore:
  \begin{equation}
  \begin{multlined}
C^D_{t+1}[j] \approx e^{\frac{\sum_{i=0}^t\frac{1}{\obs_i[j]}\log(\obs_i[j]\rate \weight_i[j])}{\sum_{i=0}^t\frac{1}{\obs_i[j]}}}\\
= e^{\frac{\sum_{i=0}^t\frac{1}{\obs_i[j]}\log(\frac{\obs_i[j]\rate \weight_i[j]C^*[j]}{C^*[j]})}{\sum_{i=0}^t\frac{1}{\obs_i[j]}}}\\
= e^{\frac{\sum_{i=0}^t\frac{1}{\obs_i[j]}\log(\frac{\obs_i[j]\rate \weight_i[j]}{C^*[j]})+\frac{1}{\obs_i[j]}\log(C^*[j])}{\sum_{i=0}^t\frac{1}{\obs_i[j]}}}\\
= e^{\frac{\sum_{i=0}^t\frac{1}{\obs_i[j]}\log(\frac{\obs_i[j]\rate \weight_i[j]}{C^*[j]})}{\sum_{i=0}^t\frac{1}{\obs_i[j]}}}e^{\frac{\log(C^*[j])\sum_{i=0}^t\frac{1}{\obs_i[j]}}{\sum_{i=0}^t\frac{1}{\obs_i[j]}}}\\
= C^*[j]e^{\frac{\sum_{i=0}^t\frac{1}{\obs_i[j]}\log(\frac{\obs_i[j]\rate \weight_i[j]}{C^*[j]})}{\sum_{i=0}^t\frac{1}{\obs_i[j]}}}.
\end{multlined}
\end{equation}

We refer the reader to the Appendix of \cite{mleflowp} for a proof of the expected value of the closed form found above since this form exactly matches the form derived for $\ourmethodclosed$.

\end{proof}

\varconvergence*
\begin{proof}
We refer the reader to the Appendix of \cite{mleflowp} for a complete proof of the theorem.
\end{proof}

\end{document}